\documentclass{sig-alternate-05-2015}

\usepackage{times}
\usepackage{graphicx} 
\usepackage{subfigure} 

\usepackage{algorithm}
\usepackage{algorithmic}
\usepackage{url}
\usepackage{hyperref}

\def\P{{ \mathrm{pr} }}

\def\C{{C_{\textrm{opt}}}}

\newcommand{\rom}[1]{\uppercase\expandafter{\romannumeral #1\relax}}

\newtheorem{remark}{Remark}
\newtheorem{theorem}{Theorem}
\newtheorem{example}{Example}
\usepackage{amsmath}
\usepackage{psfrag,epsf}
\usepackage{enumerate}
\usepackage{hyperref}
\usepackage{url} 
\usepackage{amssymb,amsfonts, subfigure,mathtools}

\begin{document}

\setcopyright{acmcopyright}

\title{Learning Time Series from Scale Information}
%
%
%
%
%

\numberofauthors{2} 
%
\author{
\alignauthor
Yuan Yang\\ 
       \affaddr{Department of Chemistry and Chemical Biology}\\
       \affaddr{Harvard University}\\
       \affaddr{12 Oxford Street, Cambridge, MA 02138 USA}\\
       \email{yuanyang@fas.harvard.edu}
\alignauthor
Jie Ding\\ 
       \affaddr{
       School of Engineering and Applied Sciences}\\
       \affaddr{Harvard University}\\
       \affaddr{33 Oxford Street, Cambridge, MA 02138 USA}\\
       \email{jieding@fas.harvard.edu}
}

\maketitle
\begin{abstract}
Sequentially obtained dataset usually exhibits different behavior at different data resolutions/scales. Instead of inferring from data at each scale individually, it is often more informative to interpret the data as an ensemble of time series from different scales. This naturally motivated us to propose a new concept referred to as the scale-based inference. The basic idea is that more accurate prediction can be made by exploiting scale information of a time series. 
We first propose a nonparametric predictor based on $k$-nearest neighbors with an optimally chosen $k$ for a single time series. 
Based on that, we focus on a specific but important type of scale information, the resolution/sampling rate of time series data. 
We then propose an algorithm to sequentially predict  time series using past data at various resolutions. 
We prove that asymptotically the algorithm produces the mean prediction error that is no larger than the best possible algorithm at any single resolution, under some optimally chosen parameters. 
Finally, we establish the general formulations for scale inference, and provide further motivating examples. 
Experiments on both synthetic and real data illustrate the potential applicability of our approaches to a wide range of time series models. 
\end{abstract}

%
%

%
%

%
%



\section{Introduction} \label{sec:intro}

Time series modeling plays a key role in understanding and predicting sequences of data that arise naturally from fields such as engineering, biology, finance, etc. The scope of time series considered in this paper are not limited to the data that are sequentially obtained in time. It may also be data sampled from space. 
In this work, we consider a novel methodology referred to as the scale-based inferences for predicting time series. The initial motivations are twofolds: first, a realistic time series usually consists of complex dynamics and cannot be modeled as a single stationary or ergodic one; second, the data size is usually massive so that offline inference 
 and model checking/selection is not computationally tractable. It is therefore preferable for an analyst to retrieve useful information of particular interest, such as trend, volatility, or one-step prediction, in a sequential manner. 
To address the aforementioned issues, we make a general assumption that, though a time series consists of distinct dynamics, such dynamics are ergodic. Thus, it is technically possible to infer the current piece of data by looking at historical data with behavioral resemblance. Here, ``dynamics'' is a general term that refers to any particular feature of time series, for example ``moment'', ``spectrum''. In a parametric setting, it can be understood as different parameters of a particular model, or different models. We will give more concrete formulations in later part of the paper. 

{\bf Related Work}:
In the time series literature, a classical way to extract the scale information such as trend and cycle is by time series decompositions\cite{cleveland1976decomposition,west1997time,baxter1999measuring}, which decomposes the data into trend/cycle plus a stationary process. 
Multi-state autoregressive models impose a parametric hierarchy on the mixture autoregressive (AR) models \cite{RACS,Allerton}. It usually assumes a Markov chain for the transition among finite number of AR states. The state can be treated as the scale information, and the data at each time step is governed by the data generating process corresponding to the particular state. 
Autoregressive Conditional Heteroskedasticity (ARCH) Models\cite{engle1982autoregressive} have been widely applied to modeling financial time series that exhibit time-varying volatility clustering. The variances that change over time can be understood as the scale information.
In the literature of statistics and signal processing, Least Absolute Shrinkage and Selection Operator (LASSO)\cite{tibshirani1996regression} and recursive $\ell_1$-regularized least squares (SPARLS)\cite{babadi2010sparls,sheikhattar2015recursive} are popular methods those aim at simultaneous model selection and inference; the two seemingly contradictory aspects are technically reconcilable, because the imposed $\ell_1$-constraint tends to select sparse number of covariates from the massive candidate set. 
They have been widely appreciated due to the robust performance and computational tractability, especially when compared with classical Akaike information criterion or Bayesian information criterion which requires combinatorial tests. 
In view of that, it can be understood as a scale-based philosophy: extract few significant covariates from a large scale while suppressing the remaining ones.    
In computer science, manifold learning methods\cite{belkin2006manifold} are based on the idea that the dimension of many datasets is only artificially high and that the data actually stay on a lower dimensional manifold.  In some manifold learning algorithms such as the locally linear embedding approach,  the modeling is based on each local information. It shares some similarities with our scale-based inference for time series in that they both aim to extract useful information inherited at each locality, which can be regarded as a scale parameter. 

{\bf Notation}: 
For two deterministic sequences $a_n,b_n$, we use $a_n = o(b_n)$ to represent $\lim_{n\rightarrow \infty} a_n/b_n = 0$. 
Let $o_p(1)$ denote any random variable that converges in probability to zero. 
We write $a_n=\Theta(b_n)$ if $c <  a_n/b_n < 1/c$ for some constant $c \neq 0$ for all sufficiently large $N$.
For a time series $x_n$ and a positive integer $d$, we let $\bar{x}_{n:n-d+1}$ denote the column vector $[x_n,x_{n-1},\cdots,x_{n-d+1}]^T$; when there is no ambiguity, we simply write it as $\bar{x}_{n}$. 
Let $|s|$ denote the Euclidean norm of a vector $s$.
Let $\mathcal{N}(\mu,V)$ denote the normal distribution with mean $\mu$ variance matrix $V$.


\section{Prediction for One Time Series} 

In general, the criteria to evaluate the performance of particularly retrieved information of interest may differ. If an analyst is interested in one-step ahead prediction, then it is natural to measure the goodness of modeling via the expected prediction error.
Suppose that the underlying data generating procedure is 
\begin{align} 
  x_n = f(x_{n-1}, x_{n-2}, \cdots,x_{n-d_0}) + \varepsilon_n  \label{eq3}
\end{align}
where $f: \mathbb{R}^d \rightarrow \mathbb{R}$ is a continuous function and $\varepsilon_n$ are independent and identically distributed noises with  mean $0$ and variance $\sigma^2$\footnote{We do not need to assume that noises are Gaussian here and also in the proof of Theorem~\ref{thm:1}.}.
The one-step prediction error (loss) is defined as $L_n = E(\hat{x}_n - x_n)^2$.
In a typical classical approach, either parametric or nonparametric, an analyst first estimates the function $\hat{f}(\cdot)$ and then apply the plug-in rule: $\hat{x}_n = \hat{f}(x_{n-1}, x_{n-2}, \cdots,x_{n-d})$.
The oracle situation is where $\hat{f} = f$ is luckily obtained and $L_n = E(\varepsilon_n^2) = \sigma^2$. Otherwise, it is straightforward to prove that $L_n > \sigma^2$. Though the lower bound $\sigma^2$ is practically not achievable, a ''good'' predictor is generally considered to be the one such that $L_n \rightarrow \sigma^2$ as $n \rightarrow \infty$ (i.e., more and more data are observed).
A major problem involved in a typical approach is that: if it is parametric (e.g., linear AR), there is a risk to mis-specify the model class (e.g., the true functional relationship is nonlinear); if it is nonparametric (e.g., spline approximation), it usually requires massive computation and it is not clear how to choose the kernels (e.g. number and spacing of knots).
We propose a simple but effective pattern-matching algorithm.
It can be categorized as a nonparametric approach, but much easier to implement than kernel-based ones since fast algorithms can be employed for searching $k$-nearest neighbors. 
Our approach is sketched in Algorithm~\ref{alg:1}.
Next, we provide theoretical analysis for its performance. We make the following assumptions. 

\begin{algorithm}[tb]
   \caption{Pattern Matching Algorithm}
   \label{alg:1}
\begin{algorithmic}[1]
\INPUT data $\{x_t:t=1,\cdots, n\}$; postulated dimension $d$; distance measure $D: \mathbb{R}^d \times \mathbb{R}^d \rightarrow R^{+} \cup \{0\}$;
\OUTPUT  predictor $\hat{x}_{n+1}$ 
   \STATE Collect 
   $ \{\bar{x}_{t_1},\cdots,\bar{x}_{t_k}\}$ which are $k$ nearest neighbors of $\bar{x}_{n}$ under distance measure $D$
   \STATE Let  
   $\hat{x}_{n+1} 
   = \sum_{i=1}^k x_{t_i+1} / k$
\end{algorithmic}
\end{algorithm}

(A1) For a constant $i$ large enough, $x_1, x_{1+i}, \cdots$ are independent. 

The assumption is applicable in practice because under mild conditions, many stationary process such as autoregression are strong-mixing \cite{athreya1986note}. In other words, as long as $i$ is allowed to diverge in any speed, $x_1, x_{1+i}, \cdots$ will become asymptotically independent.  

(A2) $f: \mathbb{R}^d \rightarrow \mathbb{R}$ is Lipschitz continuous function.

A linear function is perhaps the simplest example of Lipschitz continuous functions because of Cauchy's inequality. Thus, autoregressive model is a special case addressed here.

(A3) The sequence of random variables $\{X_n\}$ converges to a stationary distribution whose probability density function (PDF) is monotone decreasing in the tail. In other words, for a positive integer $d$, $\pi(\bar{x}_{n:n-d+1})$ is decreasing in  
  $|\bar{x}_{n:n-d+1}|$ for all $\bar{x}_{n:n-d+1} $ whose lengths are greater than a certain constant. 

A special case for (A3) to hold is the usual setting when $f(\cdot)$ is a linear function and the noises are Gaussian.
  
\begin{theorem} \label{thm:1}
  Assume (A1)-(A3) and that $D $ is the Euclidean distance, $k / n \rightarrow 0, k/\log(n) \rightarrow \infty$. 
  If $d \geq d_0$, where $d_0$ is the true dimension of the time series as defined in (\ref{eq3}), then for any bounded subspace $E \subset \mathbb{R}^{L}$, 
     for any $\bar{x}_{n+1} \in E$, $E(\hat{x}_{n+1} - x_{n+1})^2 \rightarrow \sigma^2$ a.s. with uniform convergence rate.  
\end{theorem}
\begin{proof}
  Our proof is based on constructing the nearest neighbor density estimator. Under assumption (A1), there are $n'=O(n)$ historic points independently distributed  according to  probability density function $p(\bar{x})$, the stationary distribution of $d$ consecutive points. Without loss of generality we write $n'$ as $n$, and suppose that $\bar{y}_{1},\cdots,\bar{y}_{n}$ are such points. Let $\pi_n(\bar{x})=k/(n \cdot Vol_k(\bar{x}))$ be the nonparametric density estimator at point $\bar{x}$, where $Vol_k \sim |\bar{x}-\bar{x}^{k}|^d $ is the volume of the ball with radius the distance between $\bar{x}$ and its $k$th nearest neighbor. It is well known that 
  \begin{align} \label{eq2}
    \sup_{n \rightarrow \infty} |\pi_n(\bar{x}) - \pi(\bar{x}) | = 0 \ a.s.
  \end{align}
  given $k=o(n)$ and $\log(n)=o(k)$\cite{devroye1977strong}.  
  Given any bounded subspace $E$, assumption (A2) implies that there exists a constant $\tau > 0$ such that $f(\bar{x}) \geq 2\tau $ for all $\bar{x} \in E$. In view of (\ref{eq2}), there exists $n_0$ such that for all $n>n_0$ and $\bar{x} \in E$, $Vol_k(\bar{x})<k/(n(f(\bar{x})-\tau)) \leq k/(n\tau)$. Thus, $Vol_k(\bar{x})$ converges uniformly to zero as $n$ tends to infinity.
  
  If the data points and additive noises after each $\bar{y}_i , \ i=1,\cdots,n$ are respectively denoted by $y_{i}$ and $\epsilon_i$, then the algorithms gives $\hat{x}_{n+1}=\sum_{i=1}^n y_{i}/n$ which satisfies  
  \begin{align*}
    \hat{x}_{n+1} - x_{n+1} 
    &= \frac{1}{n}\sum_{i=1}^n y_{i} - x_{n+1}\\
    &= \frac{1}{n}\sum_{i=1}^n (f(\bar{y}_i)+\epsilon_i) - (f(\bar{x})+\varepsilon_{n+1}) \\
    &= \frac{1}{n}\sum_{i=1}^n (f(\bar{y}_i) - f(\bar{x})) + \frac{1}{n}\sum_{i=1}^n \epsilon_i  - \varepsilon_{n+1} 
  \end{align*}  
  Let $K$ denote the Lipschitz constant. 
  It follows from $|\sum_{i=1}^n (f(\bar{y}_i) - f(\bar{x}))/ n | \leq K|\sum_{i=1}^n \bar{y}_i/n - \bar{x} |$ and the law of large numbers that $E(\hat{x}_{n+1} - x_{n+1})^2 \rightarrow \sigma^2 $ as $n$ goes to infinity.  
\end{proof}

\begin{remark}
Theorem \ref{thm:1} shows that for a wide range of $k$ the prediction error produced by Algo.~1 converges fast to the theoretical lower bound. The choice of $k$ that achieves the optimal rate of convergence depends on the specific data generating process which is usually unknown. But from a number of synthetic data experiments, we found that $k = \sqrt{n}$ is generally a good choice. Algo.~1 will be used in the method introduced in the next section. 	
\end{remark}

\section{Sequential Prediction  
  from Multiple Resolutions}

In this section, we study a specific scale information which is the data resolution. 
In different domains of research such as the environmental science and finance, it has been pointed out that model predictability may depend heavily on which resolution of data has been used for model fitting; in particular, high resolution (e.g., hourly data) is not necessarily better than low resolution (e.g., daily data) in terms of prediction \cite{shen2015influence,wang2013accelerating,brooks2014introductory}.
It is a natural idea to predict using each resolution, and then combine the results properly in order to achieve a more reliable prediction. In some sense, similar idea appears in model averaging methods\cite{hoeting1999bayesian}. 

We propose Algo.~\ref{alg:2} referred to as the multiple resolution based sequential prediction and learning in time series (Mr-split).
It is worth mentioning that we do not make a strict definition on ``data at resolution $r$'', as for a practitioner there are various possible definitions.
In the pseudo-code of Algo.~\ref{alg:2}, we use $y_t^{(r)}$ to denote the  pre-processed data at resolution $r$ in general. 
For example, $y_{t,n}^{(r)}$ can be $y_{t,n}^{(r)} = x_{rt}$, $t=1,\ldots,n/r$. 
We are going to provide more examples in the real data experiments. 
We note that $y_{t,n}^{(r)}$ is not necessarily a function of $x_n$, but can be exogenous variables as well.

The basic idea of Algo.~\ref{alg:2} is to linearly combine the predictor estimated from each resolution weighted by their accumulated scores, and then update the score according to some carefully chosen loss function $\ell$. The loss function takes the predictor and the revealed true value as arguments, and is used to penalize those resolution that does not perform well. 
 
Algo.~\ref{alg:2} is motivated by the exponential weight algorithm from online learning literature. However, it differs with the classic one in terms of the ultimate goal. In the exponential weight algorithm, a loss function, which is convex and bounded, is pre-defined as the performance measure. In the time series prediction, an analyst is more interested in the mean prediction error, which is usually unbounded. 
A straightforward way, as was used in (\ref{god}) of Algo.~2, is to impose a constant upper bound $C$ which we refer to as the capacity parameter.
The choice of $C$ here is essential. 
As we shall see from the proof, the optimal choice of $\eta$ is always proportional to $1/C$.
On one side, if $C$ is too large, then the score varies little at each time $n$, which means that the learning procedure is slow. On the other side, if $C$ is too small, then most of the time the loss for each resolution $r$ is constant $C$, which means that the learning procedure may not be effective in distinguishing which resolution performs better than others. 
Under mild assumptions, we prove in the following theorem that the optimal $C$ and $\eta$ should be $\Theta(N^{1/3})$ and $\eta = 1/ (C \sqrt{N}) = \Theta(N^{-5/6})$, respectively. 

\begin{algorithm}[tb]
\vspace{0.0 cm}
\caption{Multiple Resolution based Sequential Prediction and Learning in Time Series (\textbf{Mr-split})}
\label{alg:2}
\begin{algorithmic}[1]
\INPUT data $\{y^{(r)}_{t,n}:t=1,2,\ldots\}$ sequentially observed ($n=1,2,\ldots$) at each resolution $r=1,\ldots,R$; initial score for each resolution $w_{r,1}=1$; capacity parameter $C$.  

\OUTPUT  predictor $\hat{x}_{n}$ at each $n=1,2,\ldots$ 
\FOR {$n=1,2,\ldots$}
   \STATE For each resolution $r$, predict $\hat{x}_{r,n}$, from $\{y^{(r)}_{t,n}:t=1,2,\ldots\}$ using Algo.~\ref{alg:1}
   \STATE Predict $\hat{x}_{0,n} = (\sum_{r=1}^R w_{r,n} \hat{x}_{r,n}) / (\sum_{r=1}^R w_{r,n}) $
   \STATE Given the revealed true value $x_{n}$, update the score
   $w_{r,n+1} = w_{r,n} \exp[ -\eta \ \ell(\hat{x}_{n}, x_{n}) ]$, where 
   \begin{align}
   \ell(\hat{x}_{n}, x_{n}) = \max \bigl\{(x_{n}-\hat{x}_{n})^2, C \bigr\} \label{god}
   \end{align}
\ENDFOR
\end{algorithmic}
  \vspace{0.0cm}%
\end{algorithm}

Before we proceed, we make the following assumption. 

(A5)
Assume that $\{\hat{x}_{r,n}-x_n\}$ converges to a strictly stationary stochastic process with second moment, for each $ r=0,\ldots,R $, as $N$ tends to infinity. Assume that the fourth moment of $\{\hat{x}_{0,n}-x_n\}$ exists for $n=1,\ldots,N$. 

We note that $\hat{x}_{0,n}$ denotes the predictor of Mr-split at time $n$, as shown in Algo.~2.  
By ergodicity, for each $ r=0,\ldots,R$, the mean prediction error $ \sum_{n=1}^{N} E\{(\hat{x}_{r,n}-x_n)^2\} / N$ has a limit, denoted by $\sigma_r^2$, as $N$ tends to infinity. 
Suppose that $r_0 \in \{1,\ldots,R\}$ is the resolution that achieves the minimal prediction error among all the $R$ resolutions, i.e. $\sigma_{r_0}^2 \geq \sigma_{r}^2, \ r=1,\ldots,R$.

The question is whether $\sigma_{0}^2$ produced by Algo.~2 can achieve the lowest error $\sigma_{r_0}^2$.

\begin{theorem} \label{thm:online}
Under Assumption (A5), suppose that we choose the capacity
$\C = c_1 N^{1/3} $ and $\eta = c_2 N^{-5/6}$ for some constants $c_1,c_2$ that do not depend on $N$,
then Algo.~2 achieves the lowest possible prediction error as $N$ tends to infinity, i.e., 
$
\sigma_{0}^2 \leq \sigma_{r_0}^2.
$
\end{theorem}

\begin{remark}
The result shows that the predictor $\hat{x}_{0,n} $ is at least as good as the best $\hat{x}_{r,n}$ among $r=1,\ldots,R$ asymptotically. 
When $r_0 \in \{1,\ldots,R\}$ is the true data generating model, then the equality $\sigma_{0}^2 = \sigma_{r_0}^2$ holds. This is because the true model asymptotically achieves the theoretical optimal prediction error, so $\sigma_{0}^2 \geq \sigma_{r_0}^2.$
\end{remark}

\section{General Framework}

%
%

We have emphasized on a special type of scale information, the data resolution, in the last section. 
We propose a generic method for scale-based inference  in 
Algo.~\ref{alg:3}.
Below are some additional motivating examples, where $\theta_t $ denotes the scale and $\phi_t$ denotes the  parameter of the data generating process.

\begin{example} 
Consider $X_t = A_1 \sin (w_1 t) +A_2 \sin (w_2 t)+\varepsilon_t $, where $A_1 >> A_2>0, w_1 << w_2$, $\varepsilon_t$ are independent $\mathcal{N}(0,\sigma^2)$ noises. It is equivalent to $X_t \in \mathcal{N}(\theta_t+\phi_t, \sigma^2)$, where $\theta_t=A_1 \sin (w_1 t)$ can be treated as the scale as it varies much slower than $\phi_t=A_2 \sin (w_2 t)$.
\end{example}
Usually, if the parametric model is (luckily) correctly specified, then the unknown parameters ($w_1,A_1,w_2,A_2)$ can be estimated via MLE; otherwise methods such as curve fitting or time series decomposition may be used. 
We are going to apply Algo.~1 to this data in the experiments.

\begin{example}
Consider $X_t = f(X_{t-d_0+1},\cdots,X_{t-1}) + \varepsilon_t$ where $f(\cdot)$ is any linear or nonlinear continuous function.
There is no scale in this case. This example is coined for comparison with the next one.
\end{example}

\begin{example}
  Consider $X_t = f_t(X_{t-d_0+1},\cdots,X_{t-1}) + \varepsilon_t$ where $f_t \in \{f_1 , \cdots, f_s\}$, a finite set of continuous functions.
  This is a mixture model, where the scale parameter $\theta_t$ has a finite support $\theta_t \in \{1,\cdots,s\}$ and parameters $\phi_t$ depends on $\theta_t$.
 \end{example}

\begin{algorithm}[tb]
\vspace{0.0 cm}
\small
\caption{A Generic Algorithm for Scale-based Inference}
\label{alg:3}
\begin{algorithmic}[1]
\INPUT data $\{x_t:t=1,\cdots, n\}$; postulated scale parameters $\theta_n \in \Theta$ and model parameter $\phi_n \in \Phi$, where $\Theta$ and $\Phi$ are respectively the scale space and parameter space
\OUTPUT  the quantity of interest $f: (\bar{x}_{n:n-d+1}, \theta_n, \phi_n) \mapsto \mathbb{R}$  
   \STATE Infer parameter $\hat{\phi}_{n}$ and  $\hat{p} = f(\bar{x}_{n:n-d+1}; \hat{\theta}_n, \hat{\phi}_n) $ at time $n$ 
   from the past data
   \STATE Compute $\hat{\theta}_{n+1}$
\end{algorithmic}
  \vspace{0.0cm}%
\end{algorithm}

\section{Synthetic Data Experiments}

In this section, we present experimental results to demonstrate the theoretical results and the advantages of our methods on various synthetic datasets. 
The codes and related data will be made public online in the future.

\subsection{Mixed Oscillators}

\begin{figure}[htb]
\vspace{-0.2cm}
\begin{center}
\centerline{\includegraphics[width=\columnwidth]{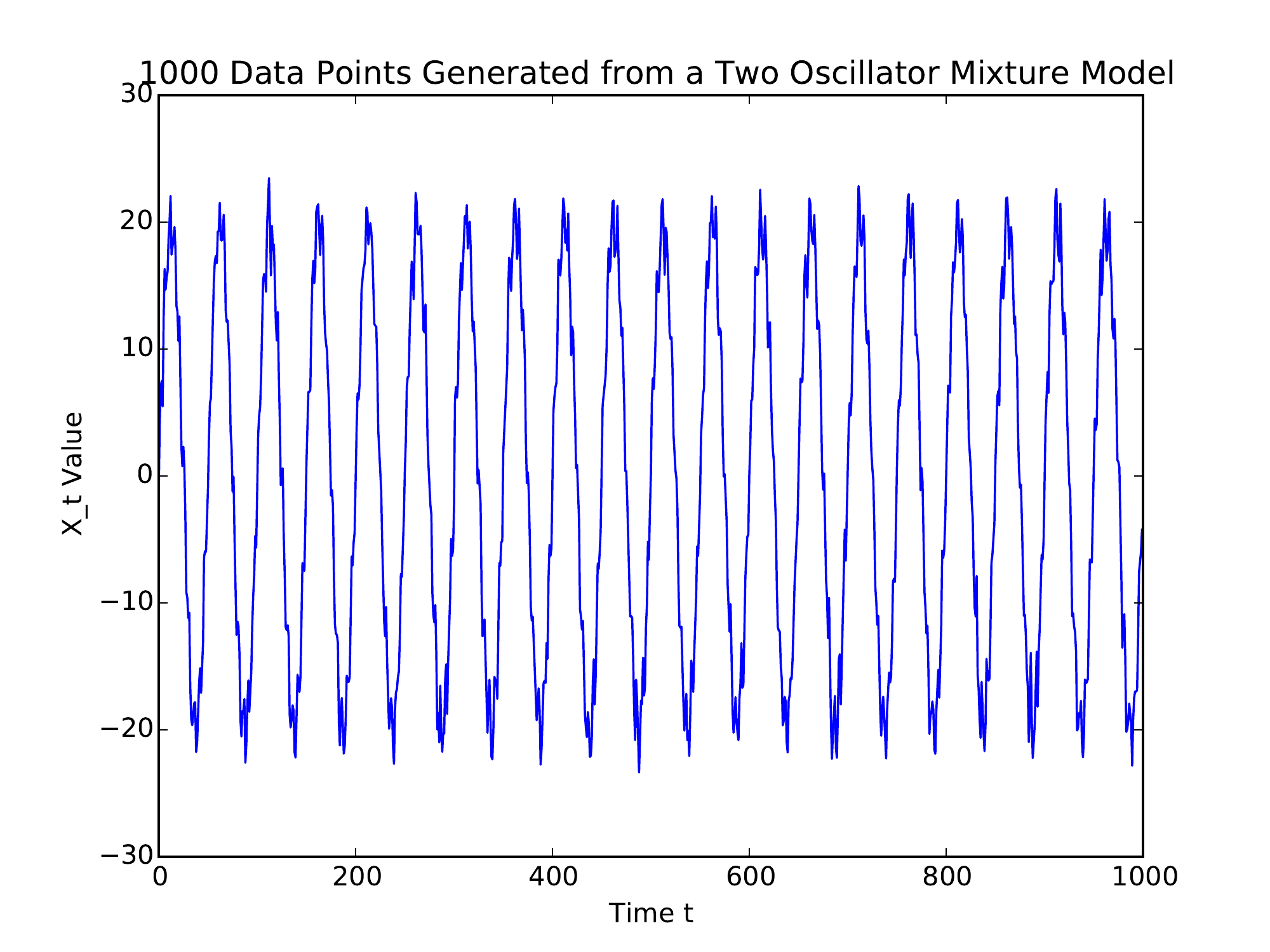}}
\caption{1000 Synthetic data generated from a two oscillators mixture model $X_t = A_1sin(w_1t)+A_2sin(w_2t)+\epsilon_t$, where $A_1=20, A_2=2, w_1=0.04\pi, w_2=0.4\pi$ and $\epsilon_t$ independently identically follows a distribution $\mathcal{N}(0,1)$.}
\label{mixoscillator_synthetic_data}
\end{center}
\vspace{-0.7cm}
\end{figure}

We first consider a case in \textbf{Example 1}. $X_t = A_1sin(w_1t)+A_2sin(w_2t)+\epsilon_t$, where $A_1=20, A_2=2, w_1=0.04\pi, w_2=0.4\pi$ and $\epsilon_t$ independently identically follows a distribution $\mathcal{N}(0,1)$. This is a typical example of fast changing data sitting on a slower changing periodic modulation. \textit{Figure}~\ref{mixoscillator_synthetic_data} shows 1000 data points generated from this mixture oscillator model. We use the first 800 data points as the training set and last 200 data points as the test set. If we luckily specify the correct parametric model and the unknown parameters $A_1,w_1,A_2,w_2$ can be estimated via methods such as MLE, then the prediction error is only from the noise part $\epsilon_t$, with a mean square error (MSE) converging to 1.\\

\begin{figure}[htb]
\vspace{-0.2cm}
\begin{center}
\centerline{\includegraphics[width=\columnwidth]{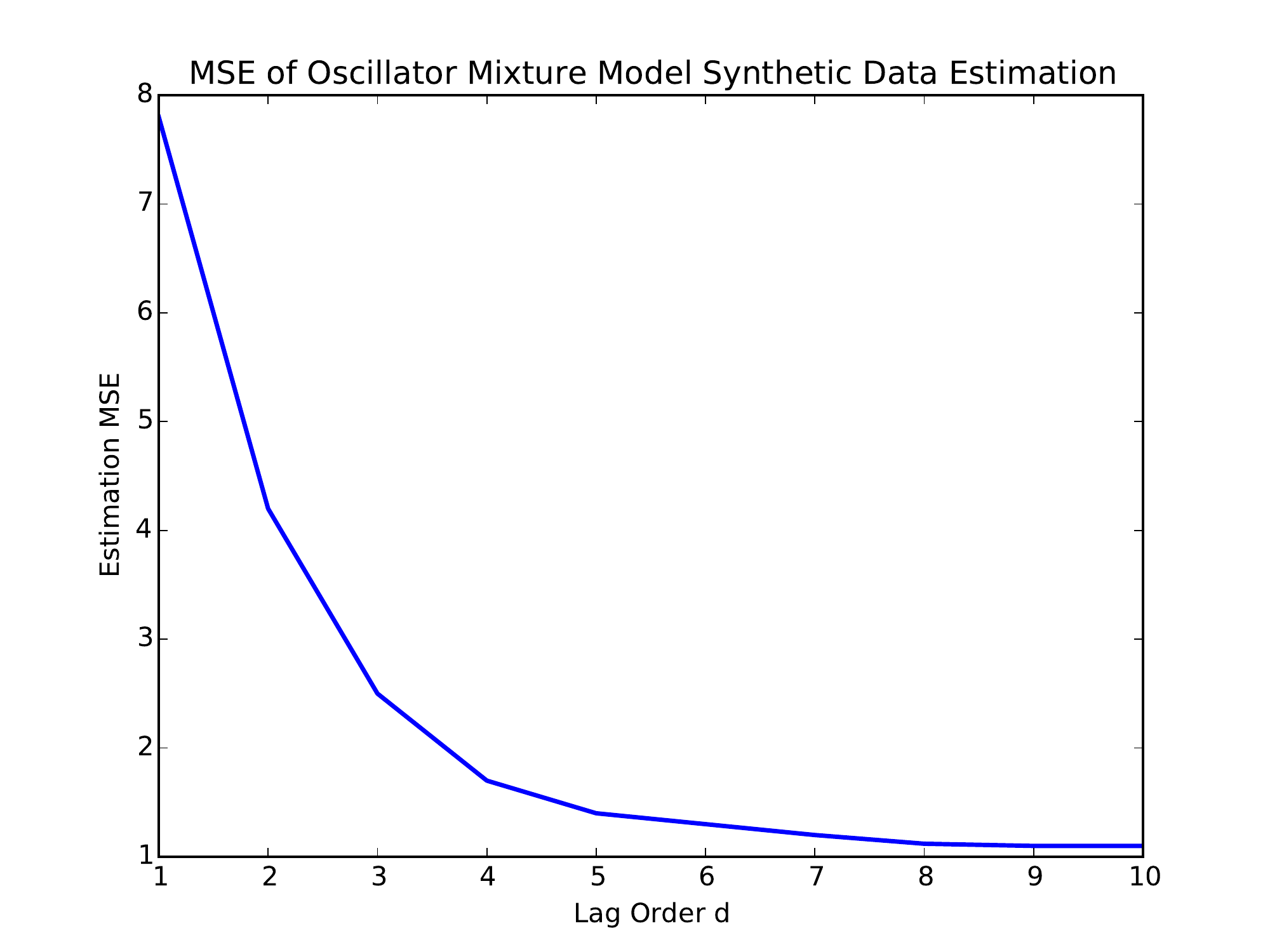}}
\caption{MSE of estimation from Algorithm 1. First 800 data points in \textit{Figure}~\ref{mixoscillator_synthetic_data} are used as the source of matching, and the last 200 data points are used for testing. Horizontal axis is the number of past points we use for pattern matching.}
\label{mixoscillator-estimator-mse}
\end{center}
\vspace{-0.7cm}
\end{figure}

Generally, we don't know the exact form of the model generating the data, specifying the model form incorrectly would lead to an inconsistent estimation.  Now we use \textbf{Algorithm 1}, the  non-parametric model to make prediction. The first 800 data points are used as the training set, and the last 200 data points are used as the test set. \textit{Figure}~\ref{mixoscillator-estimator-mse} shows the error of using different number of past points for pattern matching. The mean squared error(MSE) decreases as the number of past points increase and it converges to 1.10. We analyze why the MSE converges to 1.10 in the example: among the first 800 training data points, there are about 16 repeating patterns. So the number of repeating pattern is limited. When we find the k-nearest neighbors, we have to limit $k<16$ and we choose $k=10$ in the example. Thus the mean of 10 past predictors has variance $0.1$, which is the variance of the mean for 10 independently identically distributed $\mathcal{N}(0,1)$ noise. The variance of our estimation comes from two parts, one from the variance of estimator, and another from the noise in test data, which is $1$ in the case. Adding the variance from test data noise, the total estimation variance is 1.10. When data size increases, we will have larger $k$ when finding the k-nearest neighbors, and the variance of the estimator in \textbf{Algorithm 1} will decrease accordingly and finally converge to 0. Thus the MSE will only come from test data variance, and so MSE will converge to 1 as the data set size increases, as we proved in \textbf{Theorem 1}. This has been tested, and we omit the result here due to the page limit.\\

\subsection{Nonlinear AR(n)}

A nonlinear AR model hasn't got a good way to make prediction unless we know the exact form of the nonlinear dependence. We show here our non-parametric approach gives consistent estimation for nonlinear AR cases. We consider a case in \textbf{Example 2}: a nonlinear AR(3) sequence $X_t = 0.5X_{t-1}-0.1X_{t-2}+0.03X_{t-3}^3+\epsilon_t$, each $\epsilon_t$ independently identically follows a distribution $\mathcal{N}(0,1)$. \textit{Figure}~\ref{1ar3_synthetic_data} shows one instance of 1000 data points generated from the nonlinear AR(3) model, and we use the first 800 data points as the training set and last 200 data points as the test set. It is very difficult to tell from the data plot by eye what the form of the nonlinear dependence will be, and so we use a linear AR(n) model as a baseline without loss of generality.\\

\begin{figure}[htb]
\vspace{-0.2cm}
\begin{center}
\centerline{\includegraphics[width=\columnwidth]{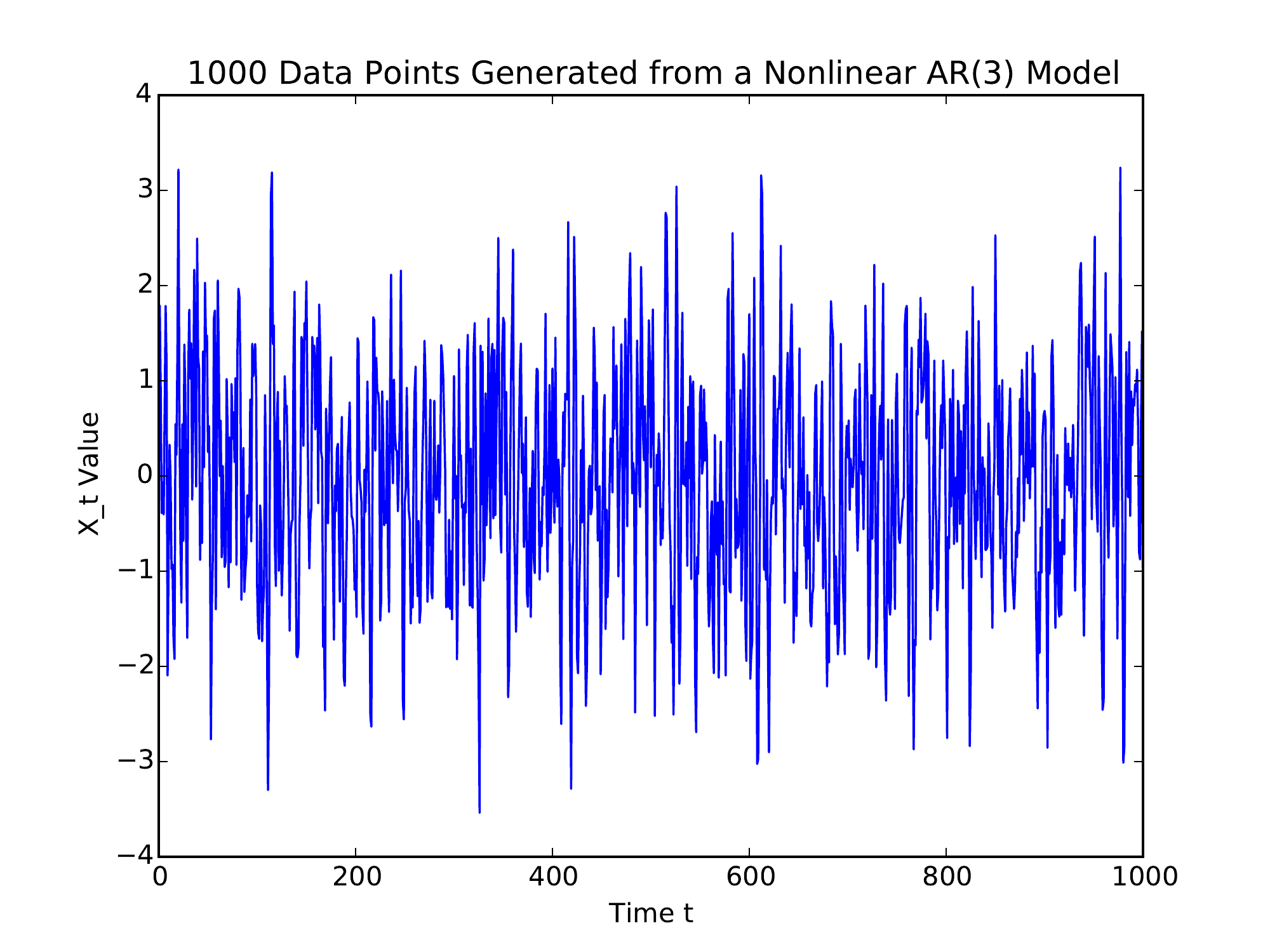}}
\caption{1000 Synthetic data generated from an AR(2) model $X_t = 0.5X_{t-1}-0.1X_{t-2}+0.03X_{t-3}^3+\epsilon_t$, each $\epsilon_t$ independently identically follows a distribution $\mathcal{N}(0,1)$.}
\label{1ar3_synthetic_data}
\end{center}
\vspace{-0.9cm}
\end{figure} 

\begin{figure}[h]
\vspace{-0.2cm}
\begin{center}
\centerline{\includegraphics[width=\columnwidth]{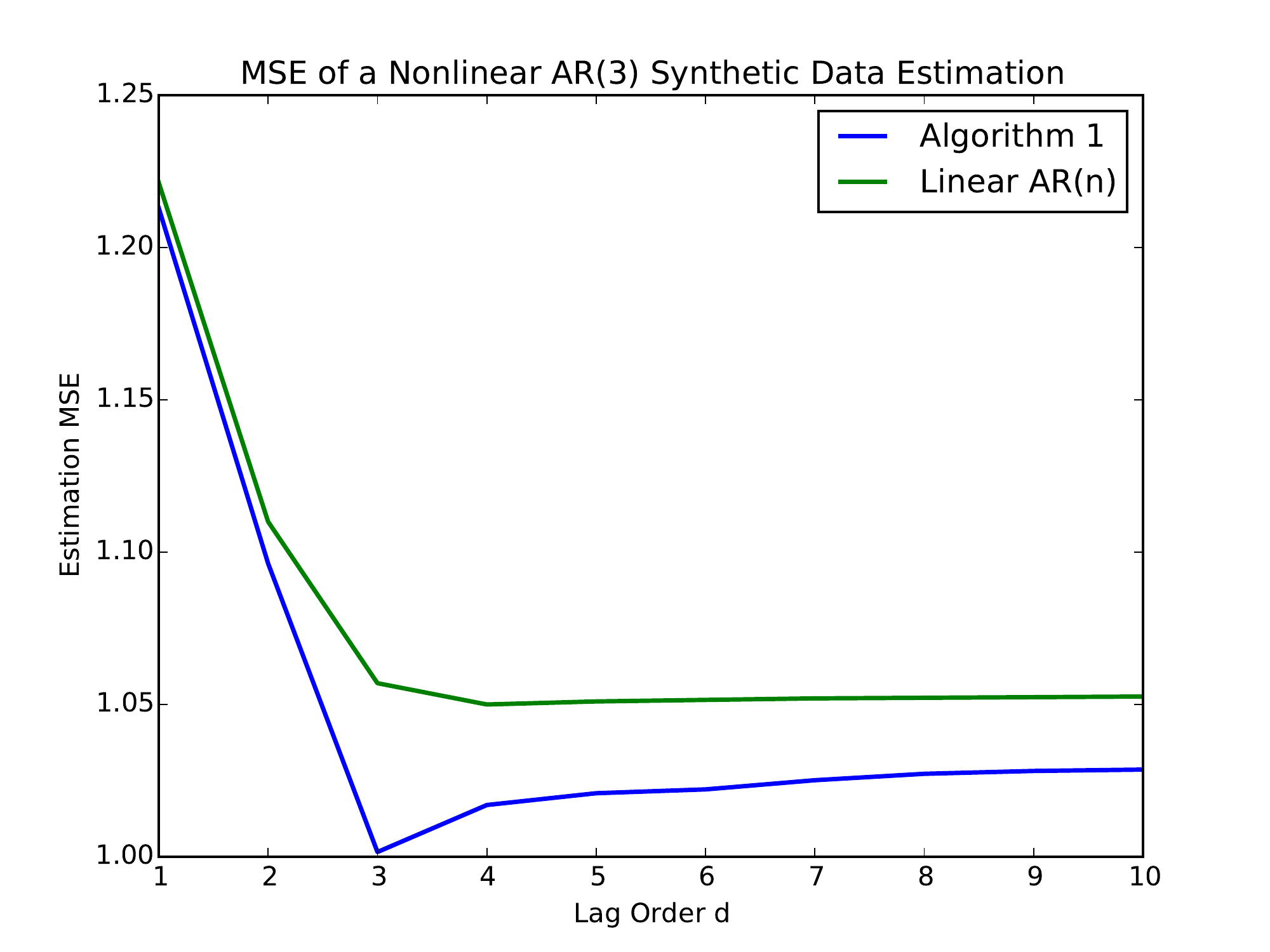}}
\caption{MSE of estimation from Algorithm 1 and linear AR(n) model. First 800 data points in \textit{Figure}~\ref{1ar3_synthetic_data} are used as the training matching set, and the last 200 data points are used for testing. Horizontal axis is the number of past points we use for pattern matching. We repeat the experiment for 100 times, and MSE at each point is the average value from 100 experiments.}
\label{1ar3-estimator-mse}
\end{center}
\vspace{-0.9cm}
\end{figure} 

For each experiment, we use the first 800 data as training data the last 200 data points as test data,  and we repeat the experiment by 100 times. \textit{Figure}~\ref{1ar3-estimator-mse} shows the average MSE for 100 experiments. For linear AR(n) model, the best MSE achieved is1.05 when n=4. For \textbf{Algorithm 1}, the one-scale non-parametric model, the error is minimized when we use the previous 3 data points as predictors. This is intuitively correct because we use a nonlinear AR(3) model to generate the data. The best mean squared error(MSE) is 1.00, and this error converges to the noise variance as we have proved in \textbf{Theorem 1}. Linear AR(n) models are a bit trivial compared to the nonlinear case. So we omit the results of linear AR(n) models.

\subsection{Multi-Resolution AR(n)}

We consider a case in \textbf{Example 3}: a mixture of three AR(2) sequences, one following $X_t = 0.65X_{t-1}-0.25X_{t-2}+\epsilon_t$, one following $X_t = -0.7X_{t-1}-0.6X_{t-2}+\epsilon_t$,  and another following $X_t = 0.6X_{t-1}-0.6X_{t-2}+\epsilon_t$, each $\epsilon_t$  independently identically follows a distribution $\mathcal{N}(0,1)$. The three sequences are mixed such that if  we sample every three data points from the mixed series, these sampled points follows one of the three one of the AR(2) time series. \textit{Figure}~\ref{3_res_synthetic_data} shows one instance of 3000 data points generated from mixing three separate AR(2) sequences, each sequence with 1000 data points. It is very difficult to tell from the data plot what the resolution of the data will be, and we will show our \textbf{Algorithm 2} framework will find out the correct resolution and converges to the correct model.\\

\begin{figure}[h]
\vspace{-0.2cm}
\begin{center}
\centerline{\includegraphics[width=\columnwidth]{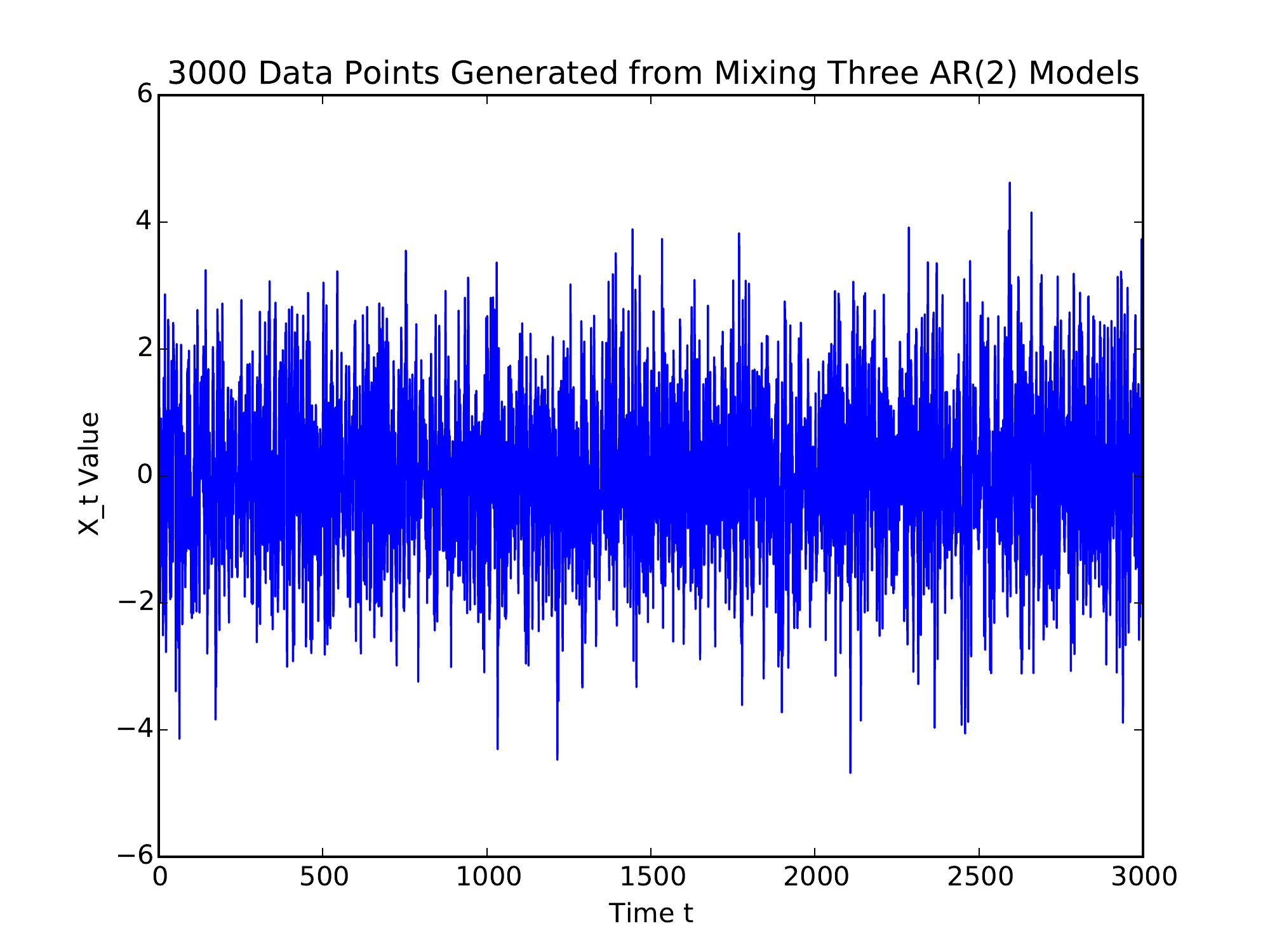}}
\caption{3000 Synthetic data by mixing three separate time series generated from three AR(2) models $X_t = 0.65X_{t-1}-0.25X_{t-2}+\epsilon_t, X_t = -0.7X_{t-1}-0.6X_{t-2}+\epsilon_t, X_t = 0.6X_{t-1}-0.6X_{t-2}+\epsilon_t$, each $\epsilon_t$ independently identically follows a distribution $\mathcal{N}(0,1)$. If we sample every three data points in the combined time series, these sampled points is one of the AR(2) time series.}
\label{3_res_synthetic_data}
\end{center}
\vspace{-0.9cm}
\end{figure} 

\begin{figure}[h]
\vspace{-0.3cm}
\begin{center}
\centerline{\includegraphics[width=\columnwidth]{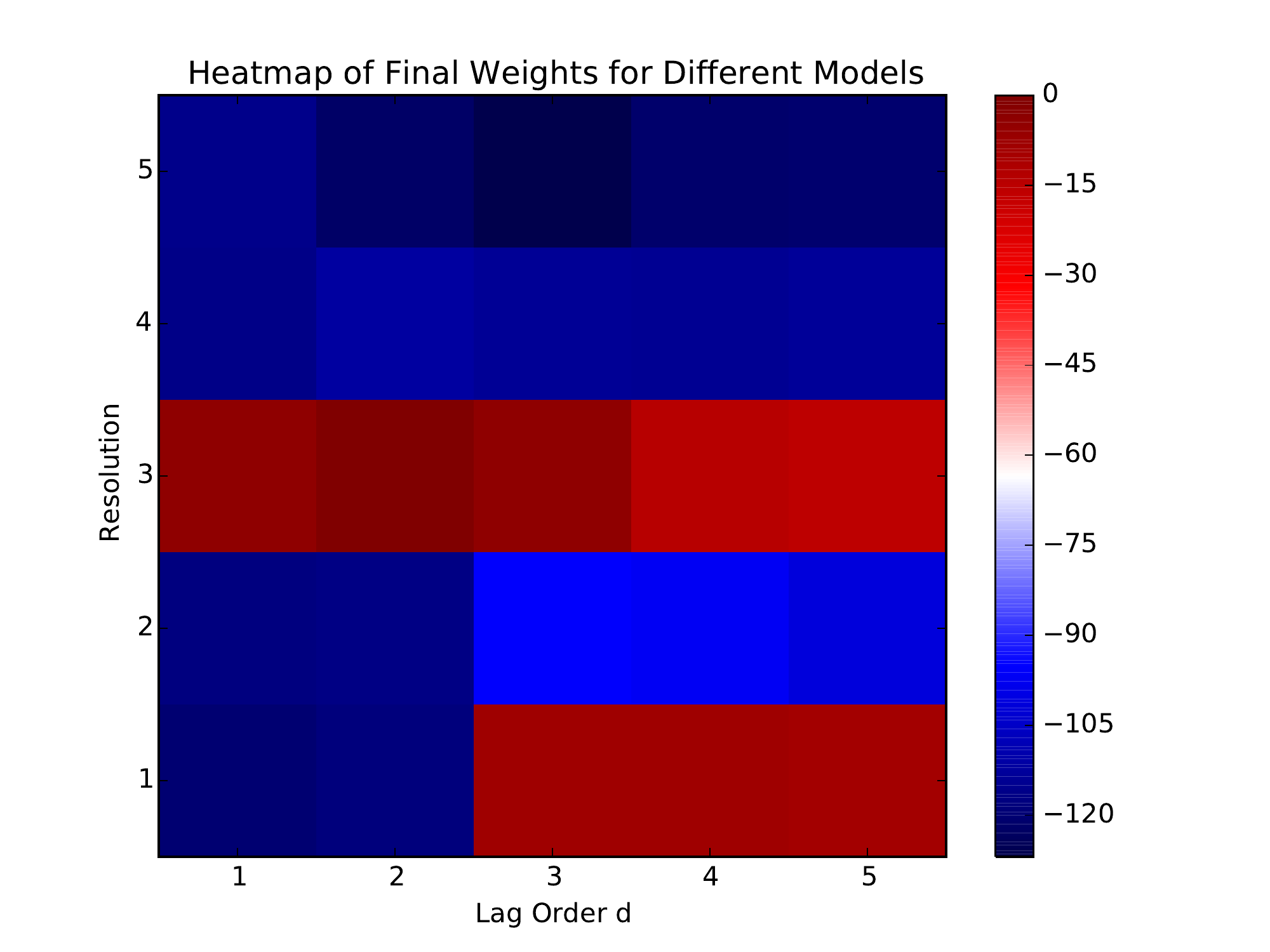}}
\caption{Heatmap shows logarithm of final weights for different models in exponential weighting algorithm. The resolution is from getting data every one data point to every five data points, and the number of predictors is also from one to five. There are 25 models in total. The weights of different models have been taken logarithm. The redder the block, the larger weight the model has. The bluer the block, the smaller weight the model has. }
\label{expweights_final}
\end{center}
\vspace{-1cm}
\end{figure}

\begin{figure}[h]
\vspace{-0.3cm}
\begin{center}
\centerline{\includegraphics[width=\columnwidth]{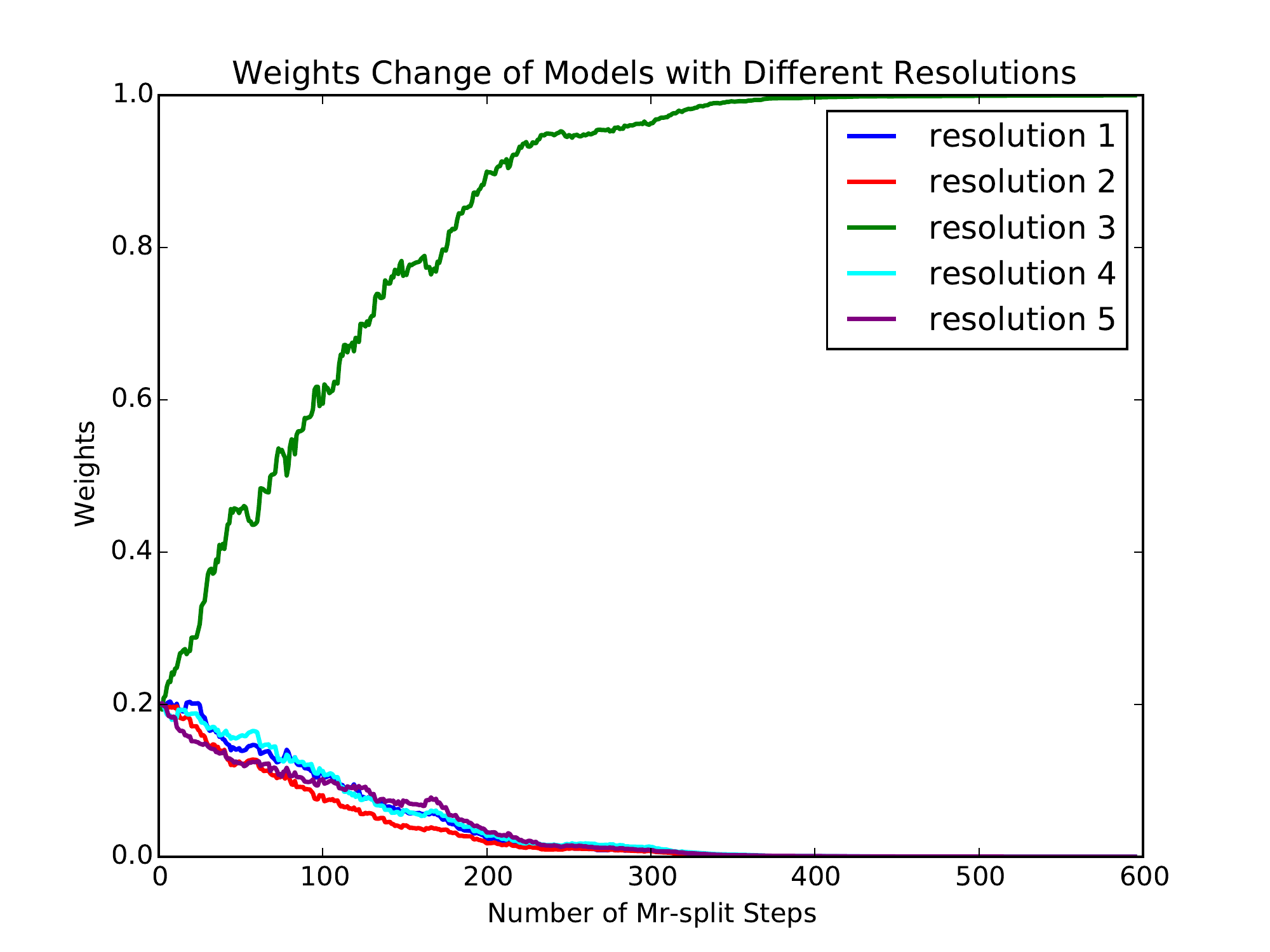}}
\caption{The weights change of different models as we make prediction on test data. First 2400 data points in \textit{Figure}~\ref{3_res_synthetic_data} are used as the training matching set, and the last 600 data points are used for testing. We keep the lag order fixed as the true value 2, and only change the resolution from every one data point to every five data points (five models in total), in order to show how the weight each model changes clearly.}
\label{expweights_change}
\end{center}
\vspace{-1cm}
\end{figure} 

\begin{figure}[h]
\vspace{-0.2cm}
\begin{center}
\centerline{\includegraphics[width=\columnwidth]{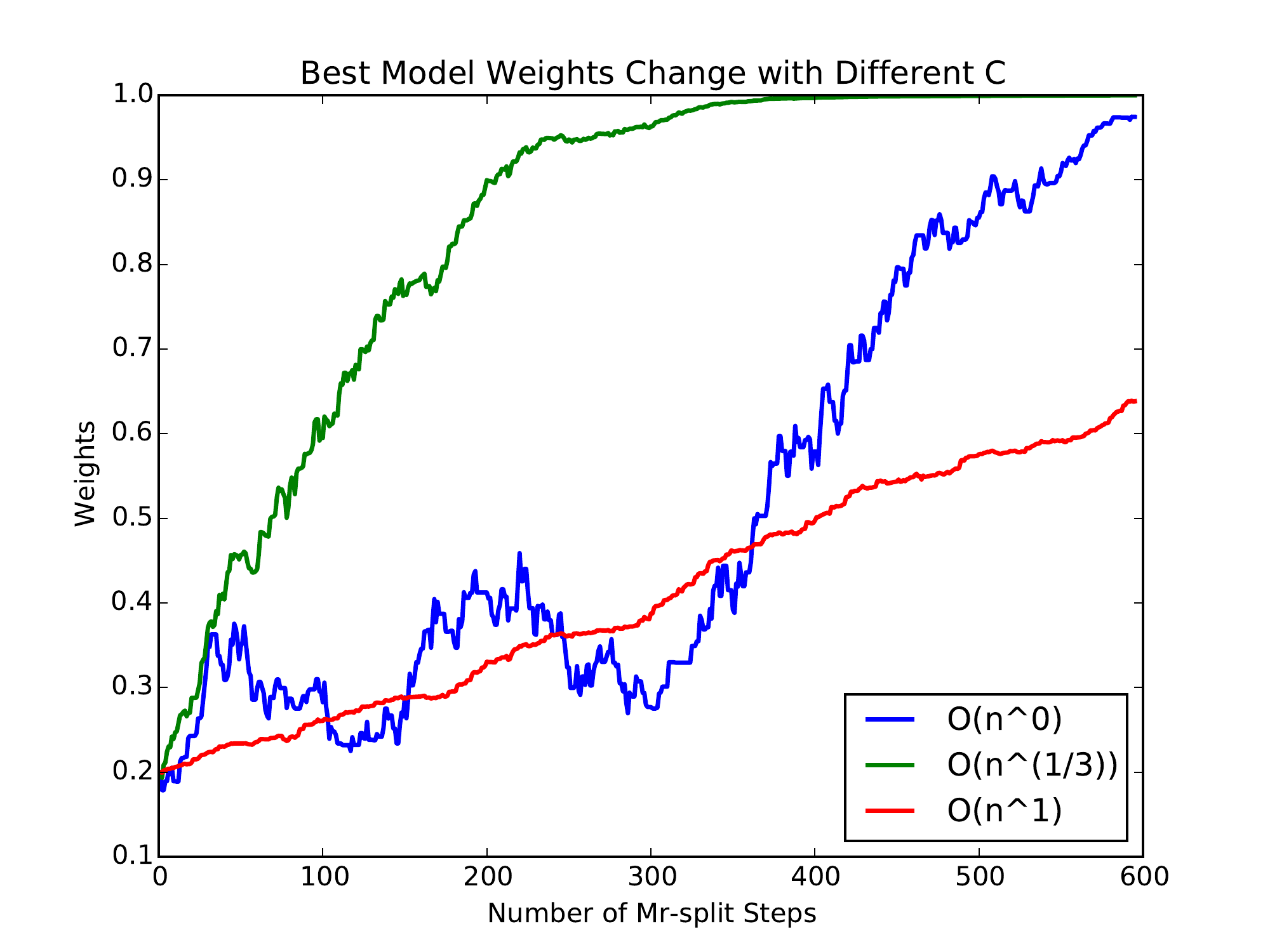}}
\caption{The weights change of the best model as we make prediction on test data by choosing different C. First 2400 data points in \textit{Figure}~\ref{3_res_synthetic_data} are used as the training matching set, and the last 600 data points are used for testing. We keep the lag order fixed as the true value 2, and only change the resolution from every one data point to every five data points (five models in total), in order to show how the weight each model changes clearly.}
\label{Best_Model_Weights_Change_DiffC}
\end{center}
\vspace{-1cm}
\end{figure} 

For each experiment, we use the first 2400 data points as the training set and last 600 data points as the test set, and we repeat the experiment by 100 times. Our algorithm assumes we don't know how many autoregressive time series are mixed or what the form the autoregressive model is or what lag order is. We scan the resolution from 1 to 5, the lag order also from 1 to 5, and there are 25 models in total. \textit{Figure}~\ref{expweights_final} shows the logarithm of final weights for each model at the end of prediction using \textbf{Mr-split} in one instance. The logarithm is taken because the relative weights for different models differ a lot in magnitude. It shows that for model with resolution 3, and using past 2 data points as predictors, has the largest weight. \textit{Figure}~\ref{expweights_change} shows the change of weights as we make predictions for one instance. For the purpose of displaying weights change clearly, we keep the lag order fixed as the true value 2, and only change the resolution from every one data point to every five data points. The correct model weight converges to 1 almost after 400 prediction steps, and weights of other models converge to 0. \textit{Figure}~\ref{Best_Model_Weights_Change_DiffC} compares the convergence rates of the best model for different capacity C values. It shows the convergence rate is optimal when choosing $C=O\left(N^{\frac{1}{3}}\right)$. The average MSE from 100 experiments is 1.00. As we have proved in \textbf{Theorem 2}, the master model converges to the true model estimation error when data size increases, which is 1 in the case, and the optimal $C=O\left(N^{\frac{1}{3}}\right)$.

\section{Real Data Experiment}
In this section, we apply our algorithms to S\&P index real data, and the advantage of our method is revealed. The stock index change could not be effectively modeled by a parametric model. We will test the effect of using multi-scale non-parametric approach to detect patterns from past data and predict future data. We adopt the quandl finance\&economics data base for the evaluation, which is available from \url{https://www.quandl.com/data/YAHOO/INDEX_SPY-SPDR-S-P-500-SPY}. 

\subsection{Daily Index Change Prediction}
Investors are more interested in the price change prediction rather than the absolute index value itself, in this subsection we apply our method to predict the relative daily price changes compared to the previous day. \textit{Figure}~\ref{sp500-dailychange} shows the relative daily change of S\&P 500 index from January, 2006 to December 2015. \\

\begin{figure}[htb]
\vspace{-0.2cm}
\begin{center}
\centerline{\includegraphics[width=\columnwidth]{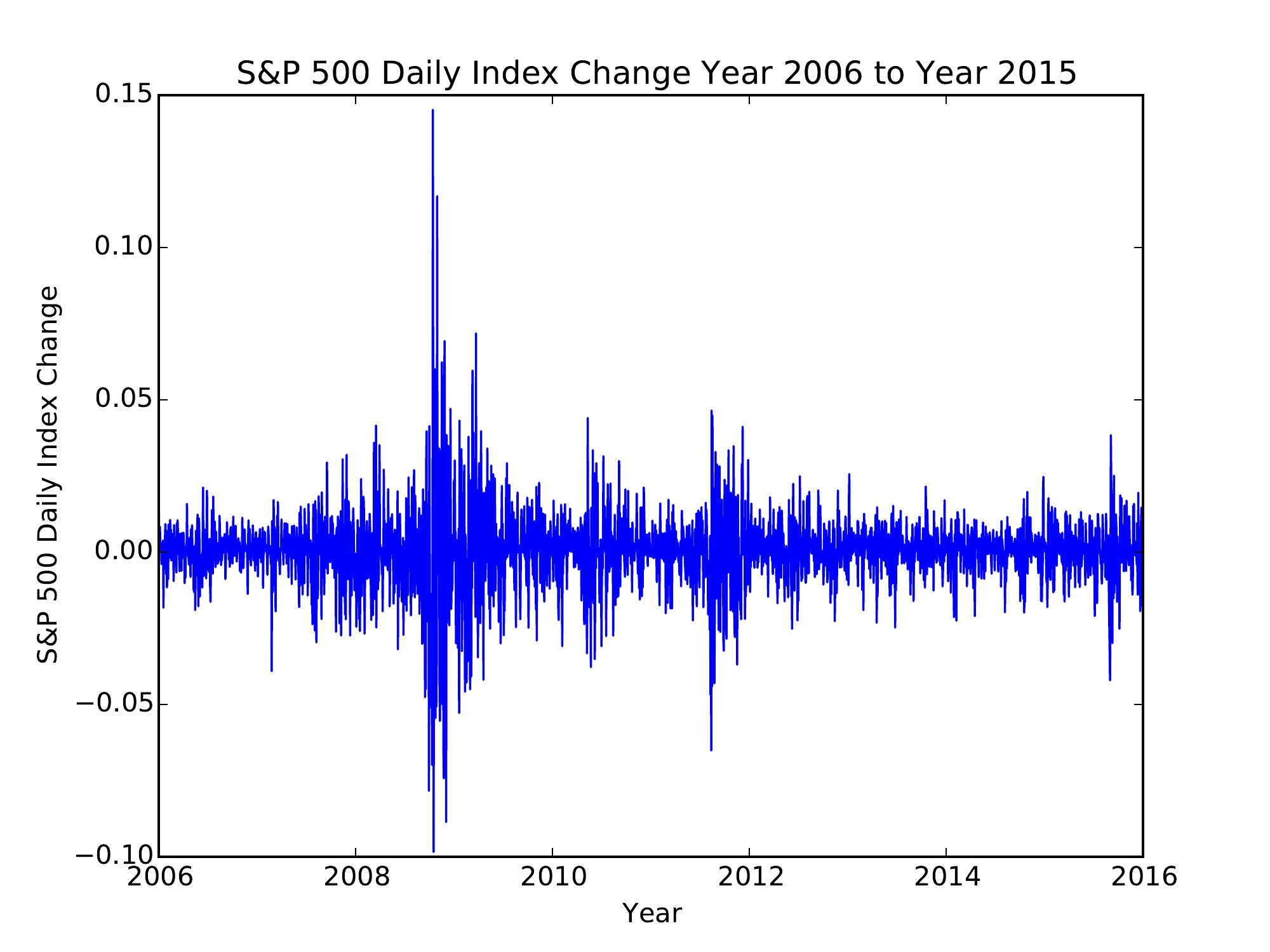}}
\caption{S\&P 500 daily index change from Year 2006 to Year 2015.}
\label{sp500-dailychange}
\end{center}
\vspace{-0.9cm}
\end{figure} 

\begin{figure}[h]
\vspace{-0.3cm}
\begin{center}
\centerline{\includegraphics[width=\columnwidth]{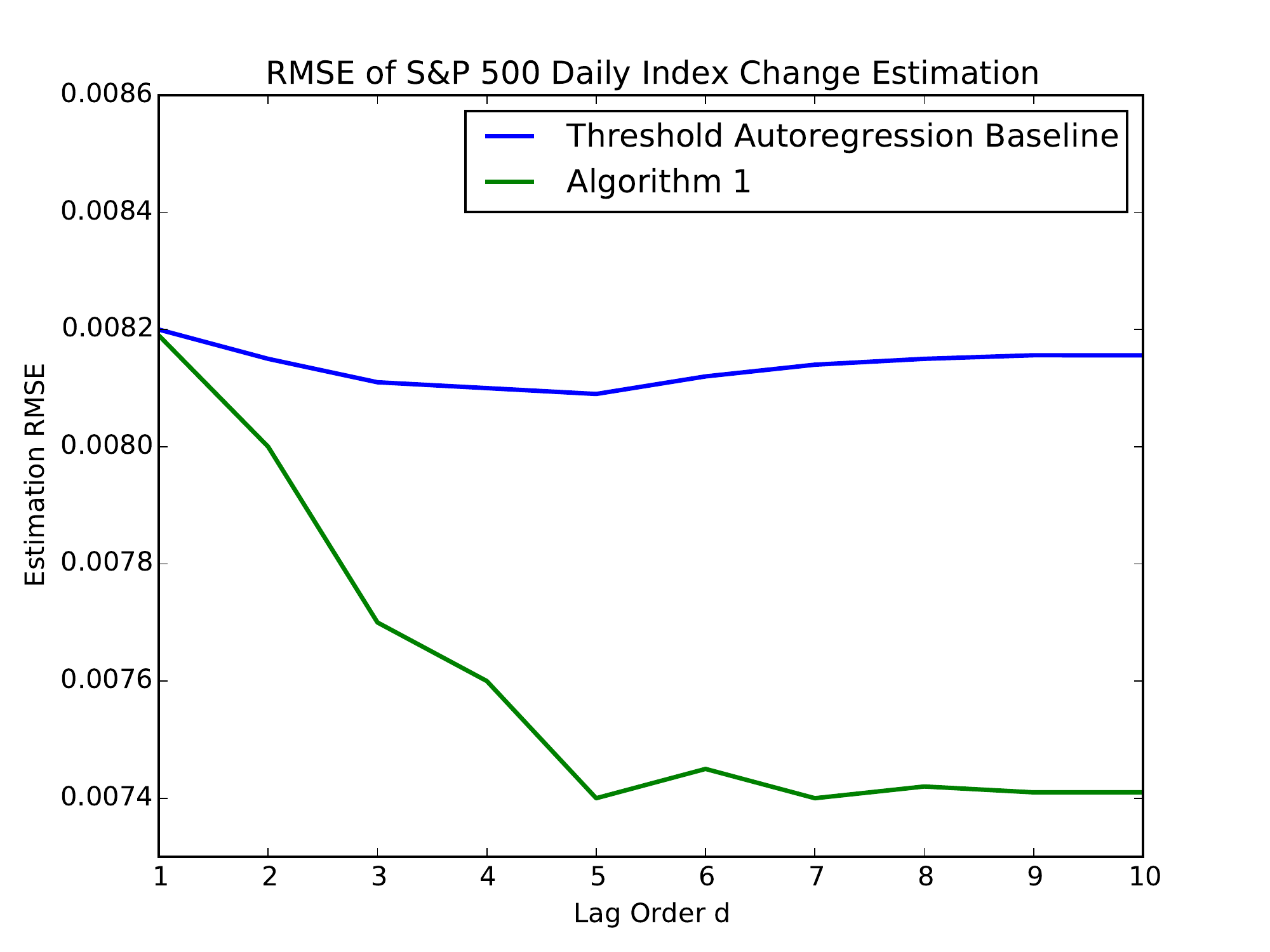}}
\caption{S\&P 500 daily index change estimation RMSE using \textbf{Algorithm 1} (green) and threshold autoregression model (blue). The horizontal axis is the number of past days we use to predict the next day. Training data: Year 2006 -- Year 2013, test data: Year 2014 -- Year 2015.}
\label{sp500-estimator-rmse}
\end{center}
\vspace{-0.9cm}
\end{figure} 

\begin{figure}[htb]
\vspace{-0.3cm}
\begin{center}
\centerline{\includegraphics[width=\columnwidth]{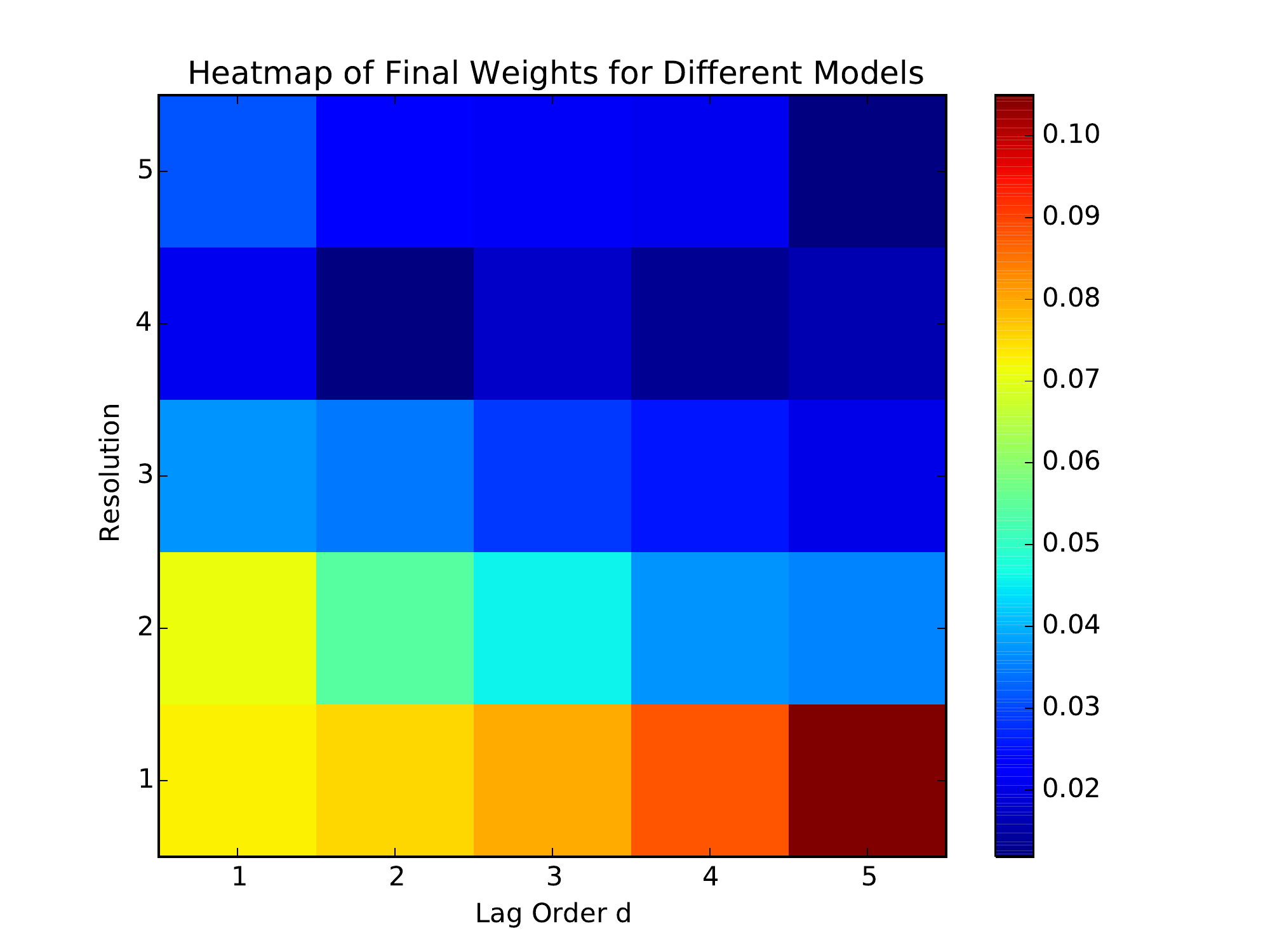}}
\caption{Heatmap shows final weights for different models in exponential weighting algorithm. The resolution is from every one data point to every five data points, and the lag order is also from one to five. There are 25 models in total. The redder the block, the larger weight the model has. The bluer the block, the smaller weight the model has. }
\label{sp500-expweights}
\end{center}
\vspace{-0.9cm}
\end{figure} 

\begin{figure}[htb]
\vspace{-0.3cm}
\begin{center}
\centerline{\includegraphics[width=\columnwidth]{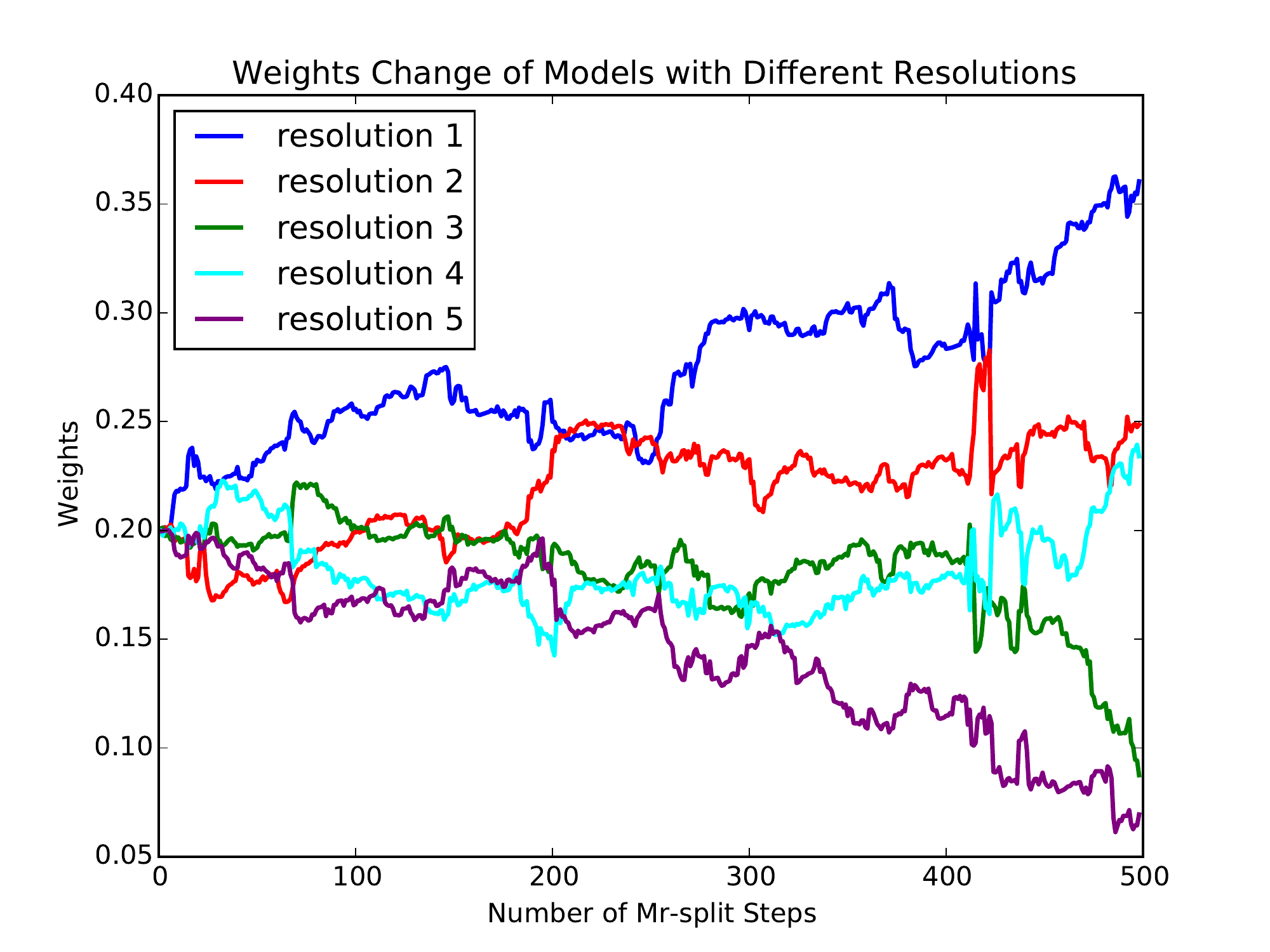}}
\caption{The weights change of different models as we make prediction on test data. We keep the lag order fixed as 5, and only change the resolution from every one data point to every five data points (five models in total), in order to show how the weight each model changes clearly.}
\label{sp500_daily_expweight_change}
\end{center}
\vspace{-0.8cm}
\end{figure} 

It is significant from the data plots that for different time period, the index change behaves very differently. In some regions, the value tends to be big and volatile, and in some other regions, the value tends to be small and stable. Threshold autoregressive (TAR) model  is widely used  to describe this phenomenon in economics. We use a one threshold autoregressive model as the baseline. Data in Year 2006 to 2013 is the training set, and data in Year 2014 and 2015 is the test set. Its prediction RMSE as a function of number of predictors is shown by the blue line in \textit{Figure}~\ref{sp500-estimator-rmse}. The horizontal axis is the number of past days we use to predict the next day. The optimal RMSE is around 0.0081 when using past 5 days data to make prediction. The green line shows the RMSE using \textbf{Algorithm 1}.  The RMSE of our non-parametric model has consistent smaller RMSE than TAR model. This is because the non-parametric approach implicitly takes possible non-linear dependence on past data into consideration, and TAR only considers the linear dependence. The RMSE is minimized to 0.0074 when using past 5 days data to make prediction. 

\textbf{Mr-split} is used to take the multiple resolution possibility into consideration. The resolution is from every one data point to every five data points, and the lag order d is also from one to five. There are 25 models in total. \textit{Figure}~\ref{sp500-expweights} shows the final weight on each model. The redder block corresponds to a larger weight to the model, and the bluer block corresponds to a smaller weight to the model. Weights are heavy for resolution 1 and this implies stock index is mostly affected by recent daily data. For resolution 1, weight is heaviest when using past 5 days data to make prediction, and so our \textbf{Mr-split} converges to the best prediction. \textit{Figure}~\ref{sp500_daily_expweight_change} shows the change of weights as we make predictions on the daily index change. For the purpose of displaying weights change clearly, we keep the lag order fixed as 5, which is the lag order minimize RMSE from \textbf{Algorithm 1}, and only change the resolution from every one data point to every five data points. The resolution 1 model, namely every one data point, gains weight as we make prediction.The optimal RMSE we get is 0.0073. This is significantly better than using traditional TAR model and even slightly better than a single best model. The reason could be some models are better at some period of time, and some better at other period of time. \\

\subsection{Monthly Average Index Change Prediction}

In this subsection, we consider estimate the monthly average index change. We define the monthly average index change as the first 20 trading days daily index change average. \textit{Figure}~\ref{sp500-monthly-avg-change} shows the relative daily change of S\&P 500 index from January, 2006 to December 2015. We are interested in if a low resolution prediction, namely using past month average index change data to make prediction better, or if a high resolution prediction, namely, for example, we predict everyday index changes in a month, and use their average as monthly average prediction better. A higher resolution prediction requires more information from data to make prediction and a higher resolution prediction appears better intuitively, but our algorithm indicates this is not necessarily true.\\

\begin{figure}[htb]
\vspace{-0.2cm}
\begin{center}
\centerline{\includegraphics[width=\columnwidth]{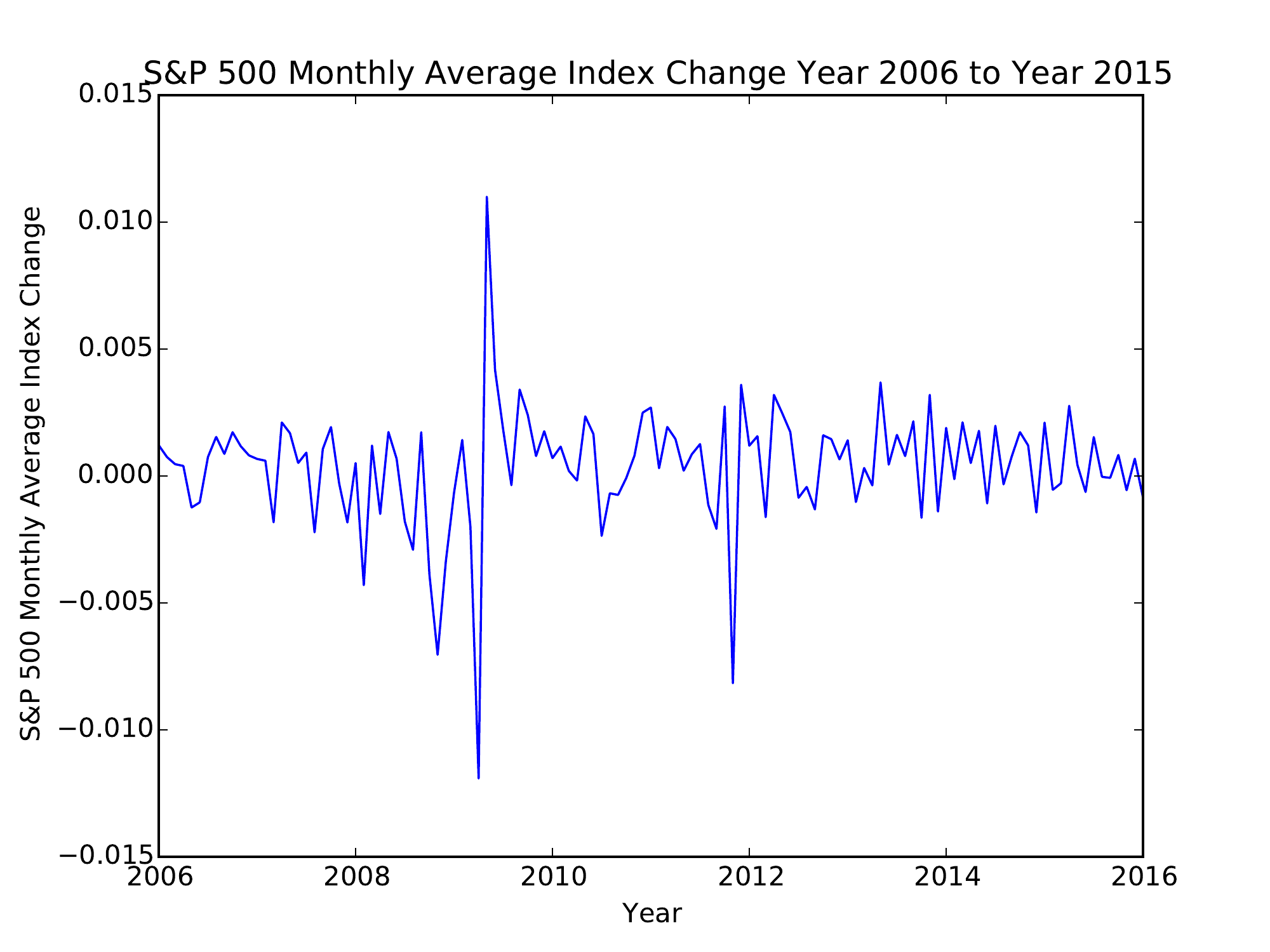}}
\caption{S\&P 500 monthly average index change from Year 2006 to Year 2015.}
\label{sp500-monthly-avg-change}
\end{center}
\vspace{-1cm}
\end{figure}

\begin{figure}[htb]
\vspace{-0.2cm}
\begin{center}
\centerline{\includegraphics[width=\columnwidth]{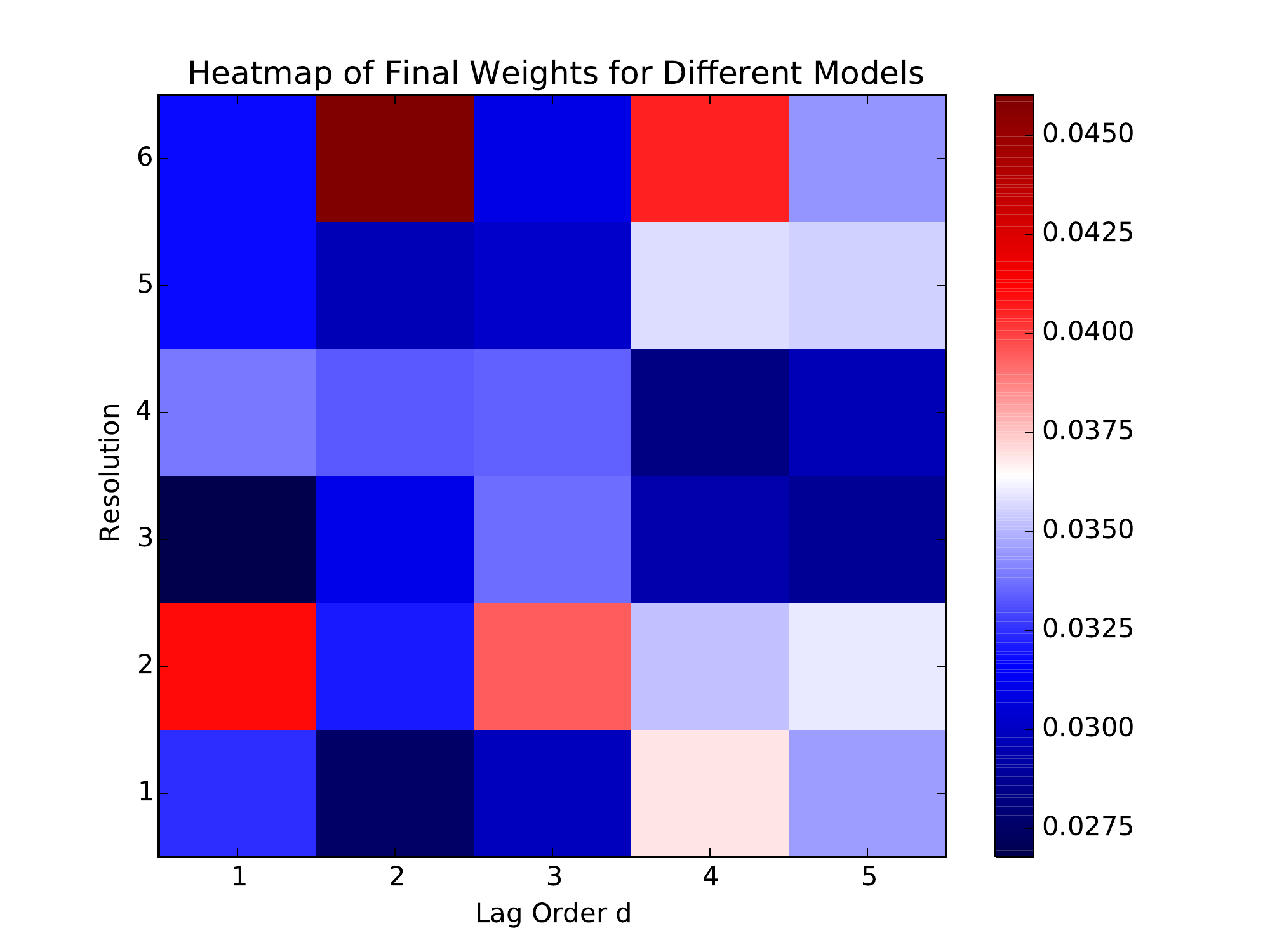}}
\caption{Heatmap shows final weights for different models in exponential weighting algorithm. The resolution includes 1-day data, 2-day average data, 4-day average data, weekly (5-day) average data, 10-day average data, monthly (20-day) average data, six resolution levels (marked as resolution 1 to 6 in the figure), and the lag order is also from one to five. There are 30 models in total. The redder the block, the larger weight the model has. The bluer the block, the smaller weight the model has. }
\label{sp500-expweights2}
\end{center}
\vspace{-1cm}
\end{figure} 

\begin{figure}[h]
\vspace{-0.2cm}
\begin{center}
\centerline{\includegraphics[width=\columnwidth]{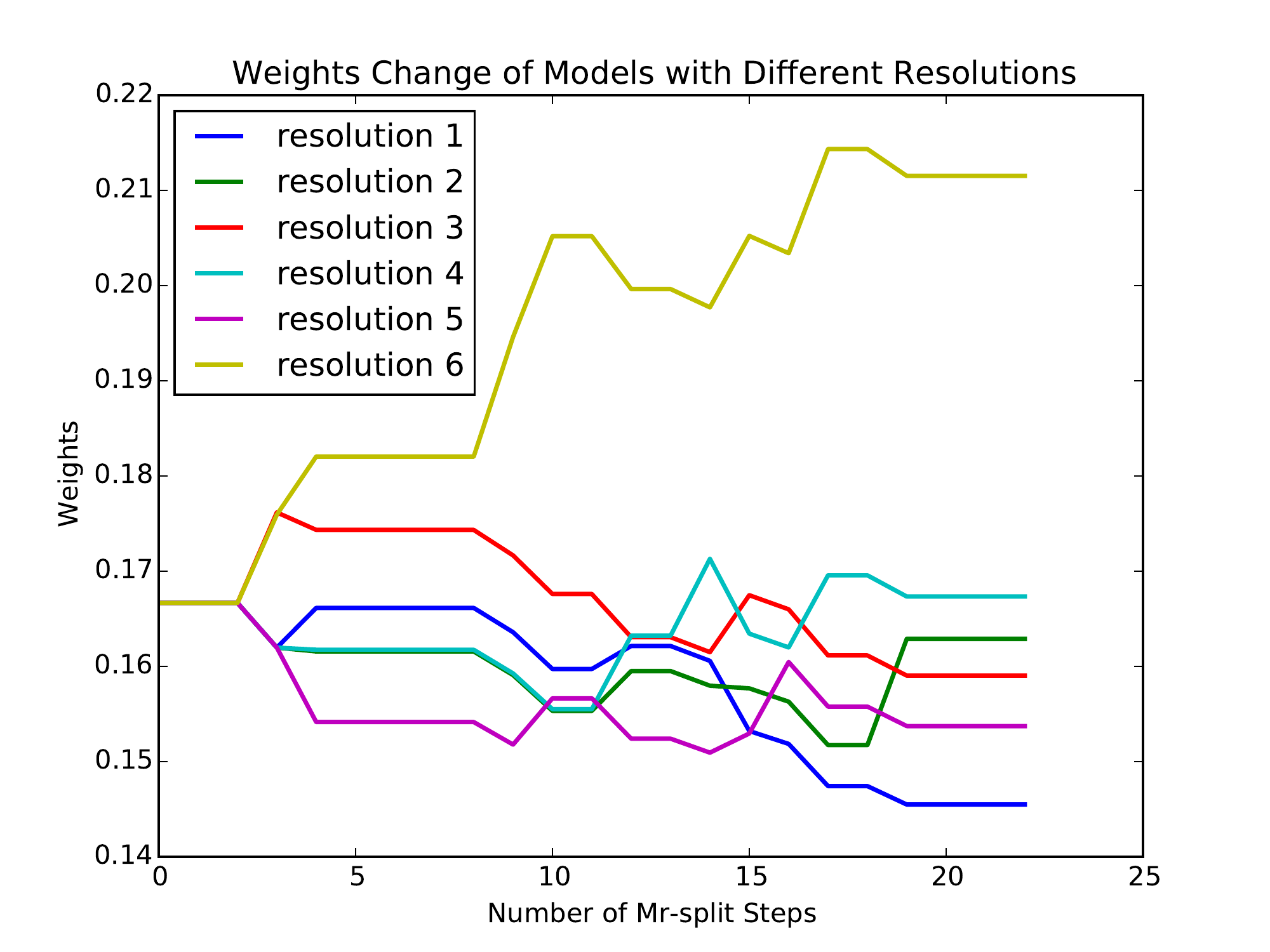}}
\caption{The weights change of different models as we make prediction on test data. We keep the lag order fixed as 2, and only change resolutions: 1-day data, 2-day average data, 4-day average data, weekly (5-day) average data, 10-day average data, monthly (20-day) average data (six resolution levels in total), in order to show how the weight of each model changes clearly.}
\label{sp500_monthavg_expweights_change}
\end{center}
\vspace{-1cm}
\end{figure} 

We include 1-day data, 2-day average data, 4-day average data, weekly (5-day) average data, 10-day average data, monthly (20-day) average data, six resolution levels in our algorithm. The lag order d is from one to five. There are 30 models in total. \textit{Figure}~\ref{sp500-expweights2} shows the final weight on each model. It is clear from the weight that model with resolution label 6, lag order 2 has the largest weight, namely using monthly average data and making prediction based on past two monthly average data. For model with the highest resolution (labelled as Resolution 1), they get lower weight. \textit{Figure}~\ref{sp500_monthavg_expweights_change} shows the change of weights as we make predictions on the monthly average index change. For the purpose of displaying weights change clearly, we keep the lag order fixed as 2, and only change the resolution from every one data point to every five data points. The resolution 6 model, namely every one data point, gains weight as we make prediction.The optimal RMSE we get is 0.0011. The optimal RMSE got from TAR is 0.0013. \textbf{Mr-split} gives better prediction compared to traditional methods.\\

\section{Conclusion}
We design a non-parametric approach for time series inference based on resolutions/scales information. A non-parametric pattern matching approach that gives consistent prediction on a general autoregressive model is proposed. We also proposed a sequential prediction algorithm that gives good estimation for complex models by combining models from multiple resolutions. Experiments on both synthetic data and real data show that the approach is applicable to complex time series data inference.

\section{Acknowledgements}

The authors are especially grateful to Professor Finale Doshi-Velez for carefully reading the draft and providing valuable suggestions in terms of how to improve the quality of the paper. We also thank teaching fellows and classmates in the machine learning course at Harvard University for helpful discussions.

\appendix
\section{Proof of Theorem 2}

For $n=1,\ldots,N$, resolution $r$ gives predictor $\hat{x}_{r,n} \in \mathbb{R}$. The true value $x_n$ is then revealed. The $r$th resolution suffers loss $\ell(\hat{x}_{r,n},x_n)$, where $\ell(\cdot)$ is defined as $\ell(\hat{x},x) = \min\{(\hat{x}-x)^2, C\}$ for some properly chosen constant $C>0$.
Our goal is to achieve the optimal prediction error $\sigma^2$ for the combined predictor $\hat{x}_{0,n}$. Suppose that the $r_0$th resolution is the true data generating resolution. We define the cumulated loss $L_{0,N}$ and $L_{r_0,N}$ as 
$$
L_{0,N} = \sum\limits_{n=1}^N \ell(\hat{x}_{0,n},x_n) , \ 
L_{r_0,N} = \ell(\hat{x}_{r_0,n},x_n).
$$
Note that the weights can be defined equivalently by letting $w_{r,1}=1$ and $w_{r,n+1} =w_{r,n} e^{-\eta \ell(\hat{x}_{r,n},x_n)}$. 
Define $W_n = \sum_{r=1}^R w_{r,n}$. It follows that $W_1=N$ and 
\begin{align} 
&\log \frac{W_{N+1}}{W_1} 
= \log \sum\limits_{r=1}^R w_{r,N+1} - \log R 
= \log \sum\limits_{r=1}^R e^{-\eta L_{r,N}} - \log R  \nonumber \\
&\geq \log \biggl( e^{-\eta L_{r_0,N}} \biggr)- \log R 
=-\eta L_{r_0,N} - \log R  \label{eq101}
\end{align}
We define $A_n$ to be the event that $ \max_{r=1,\ldots,R} \ell(\hat{x}_{r,n},x_n)<C $, $A_n^c$ the complement of $A_n$, and $\tilde{A}_N = \cap_{n=1}^N A_n$. 
Since 
\begin{align*}
\log \frac{W_{n+1}}{W_n}=\log \frac{\sum\limits_{r=1}^R w_{r,n+1}}{\sum\limits_{r=1}^R w_{r,n}} 
= \log \frac{\sum\limits_{r=1}^R w_{r,n}e^{-\eta L_{r,n}} }{\sum\limits_{r=1}^R w_{r,n}}
\end{align*}
we obtain
\begin{align}
E \log \frac{W_{n+1}}{W_n}  
&=E \biggl( \log \frac{W_{n+1}}{W_n} | \tilde{A}_n  \biggr) P(\tilde{A}_n) \nonumber \\
&+E \biggl( \log \frac{W_{n+1}}{W_n} | \tilde{A}_n^c  \biggr) P(\tilde{A}_n^c) \nonumber \\ 
& \leq E \biggl( \log \frac{W_{n+1}}{W_n} | \tilde{A}_n \biggr)  \label{eq:104}
\end{align}
where we have used $P(\tilde{A}_n) \leq 1$ and 
$W_{n+1} / W_n \leq 1$. 

Conditioning on $\tilde{A}_n$, the value of $\log (W_{N+1}/W_1)$ can be bounded by 
\begin{align} 
 \log \frac{W_{n+1}}{W_n}  
&\leq  -\eta \frac{\sum\limits_{r=1}^R w_{r,N} \ell(\hat{x}_{r,n},x_n) }{\sum\limits_{r=1}^R w_{r,N}} + \frac{\eta^2 C^2}{8}  \label{eq105} \\
&\leq  -\eta  \ell(\hat{x}_{0,n},x_n)   + \frac{\eta^2 C^2}{8} \label{eq:102}
\end{align}
where we have applied Hoeffding's inequality\footnote{Hoeffding's inequality states that $\log E(e^{sY}) \leq sEY+s^2(a-b)^2/8$ for any random variable $Y \in [a,b]$.} to (\ref{eq101}), and Jensen's inequality to the first part of (\ref{eq105}) since $\ell$ is convex in its first argument.
Bringing (\ref{eq:102}) into (\ref{eq:104}) for $n=1,\ldots,N$, we obtain
\begin{align}
E\biggl( \log \frac{W_{N+1}}{W_1}  \biggr)
&=E\biggl( \sum\limits_{n=1}^N \log \frac{W_{n+1}}{W_n} \biggr) \nonumber \\
&\leq 	
 -\eta \sum\limits_{n=1}^N \ell(\hat{x}_{0,n},x_n)  
 + \frac{\eta^2 C^2 N}{8} \label{eq110}
\end{align}
Combining Inequalities (\ref{eq110}) and (\ref{eq101}), we get 
\begin{align}
E (L_{0,N}) 
&\leq E( L_{r_0,N}) + \frac{\log R}{\eta} + \frac{\eta C^2 N}{8}   \nonumber \\
&
= E( L_{r_0,N}) + \sqrt{\frac{N \log R}{2}}C  \label{eq114}
\end{align}
where the last equality is achieved when 
\begin{align}	
\eta = \frac{1}{C}\sqrt{ \frac{8 \log R}{N} }.
\end{align}

What we are truly interested in is whether 
$$\frac{1}{N} \sum\limits_{n=1}^N (\hat{x}_0-x)^2 $$ converges to the optimum $\sigma^2$.

From our assumption and the Birkhoff Ergodic Theorem, 
\begin{align}
\frac{E(L_{0,N})}{N} 
=\frac{1}{N}\sum\limits_{n=1}^N E\biggl[ \min\{ (\hat{x}_{0,n}-x_n)^2, C\} \biggr]  \label{eq115}
\end{align}
converges to $E[ \min\{ Y_{0}, C \} ]$,
where $Y_{0}$ denotes the stationary distribution of $(\hat{x}_{0,n}-x_{n})^2$, as $N$ tends to infinity.
We use $\textbf{1}_{E}$ as the indicator random variable of event $E$. Since 
\begin{align}
E[ \min\{ Y_{0}, C \} ] 
&=	E( Y_{0} \textbf{1}_{Y_{0}<C} + C \textbf{1}_{Y_{0}\geq C} ) \nonumber \\
&=E( Y_{0}) - E( Y_{0} \textbf{1}_{Y_{0}>C} ) + C \ \P(Y_0 \geq C)  \nonumber \\
&\geq E( Y_{0}) - \sqrt{ E (Y_{0}^2) E( \textbf{1}_{Y_{0}>C}) }
 \label{eq111} \\
&\geq E( Y_{0}) - \sqrt{ \frac{E (Y_{0}^2) \ \P(Y_0>C) }{C} } \label{eq112} \\
&\geq E( Y_{0}) - \sqrt{ \frac{E (Y_{0}^2) }{C} }\label{eq113} 
\end{align}
where Inequality (\ref{eq111}) is from Cauchy's inequality and $\P(Y_0 \geq C) \geq 0$, (\ref{eq112}) is from Markov's inequality, and (\ref{eq113}) is from $\P(Y_0 \geq C) \leq 1$.
Combining the results from (\ref{eq113}), (\ref{eq114}) and (\ref{eq115}),  we obtain 
\begin{align}
E( Y_{0}) &\leq 	\frac{E( L_{r_0,N})}{N}+ \sqrt{ \frac{E (Y_{0}^2) }{C} }  + \sqrt{\frac{ \log R}{2N}}C + o(1) \nonumber \\
&\leq E\{(\hat{x}_{r_0}-x_n)^2\} + \sqrt{ \frac{E (Y_{0}^2) }{C} } 
+ \sqrt{\frac{\log R}{2N}}C + o(1) \nonumber \\
&= \sigma_{r_0}^2 + \frac{3}{2} \{E(Y_0^2)\}^\frac{1}{3} \bigl( 2 \log R \bigr)^{\frac{1}{6}} N^{-\frac{1}{6}} + o(1) \label{eq116} \\
&=\sigma_{r_0}^2 + o(1) \nonumber
\end{align}
as $N$ tends to infinity,
where the equal sign in (\ref{eq116}) is achieved at 
\begin{align}
\C=\biggl( \frac{E(Y_0^2)}{2 \log R} N \biggr) ^{\frac{1}{3}}   \label{eq:opt}
\end{align}	
This further implies that $\sigma_0^2 = E( Y_{0}) \leq \sigma_{r_0}^2$.

\bibliography{multiscale_learning} 

\begin{thebibliography}{10}

\bibitem{athreya1986note}
K.~B. Athreya and S.~G. Pantula.
\newblock A note on strong mixing of arma processes.
\newblock {\em Statistics \& probability letters}, 4(4):187--190, 1986.

\bibitem{babadi2010sparls}
B.~Babadi, N.~Kalouptsidis, and V.~Tarokh.
\newblock Sparls: The sparse rls algorithm.
\newblock {\em Signal Processing, IEEE Transactions on}, 58(8):4013--4025,
  2010.

\bibitem{baxter1999measuring}
M.~Baxter and R.~G. King.
\newblock Measuring business cycles: approximate band-pass filters for economic
  time series.
\newblock {\em Review of economics and statistics}, 81(4):575--593, 1999.

\bibitem{belkin2006manifold}
M.~Belkin, P.~Niyogi, and V.~Sindhwani.
\newblock Manifold regularization: A geometric framework for learning from
  labeled and unlabeled examples.
\newblock {\em The Journal of Machine Learning Research}, 7:2399--2434, 2006.

\bibitem{brooks2014introductory}
C.~Brooks.
\newblock {\em Introductory econometrics for finance}.
\newblock Cambridge university press, 2014.

\bibitem{cleveland1976decomposition}
W.~P. Cleveland and G.~C. Tiao.
\newblock Decomposition of seasonal time series: A model for the census x-11
  program.
\newblock {\em Journal of the American statistical Association},
  71(355):581--587, 1976.

\bibitem{devroye1977strong}
L.~P. Devroye and T.~Wagner.
\newblock The strong uniform consistency of nearest neighbor density estimates.
\newblock {\em The Annals of Statistics}, pages 536--540, 1977.

\bibitem{Allerton}
J.~Ding, M.~Noshad, and V.~Tarokh.
\newblock Data-driven learning of the number of states in multi-state
  autoregressive models.
\newblock {\em 53rd Annual Allerton Conference on Communication, Control, and
  Computing}, 2015.

\bibitem{RACS}
J.~Ding, M.~Noshad, and V.~Tarokh.
\newblock Sequential learning of multi-state autoregressive time series.
\newblock {\em RACS Proceedings of the 2015 Conference on research in adaptive
  and convergent systems}, 2015.

\bibitem{engle1982autoregressive}
R.~F. Engle.
\newblock Autoregressive conditional heteroscedasticity with estimates of the
  variance of united kingdom inflation.
\newblock {\em Econometrica: Journal of the Econometric Society}, pages
  987--1007, 1982.

\bibitem{hoeting1999bayesian}
J.~A. Hoeting, D.~Madigan, A.~E. Raftery, and C.~T. Volinsky.
\newblock Bayesian model averaging: a tutorial.
\newblock {\em Statistical science}, pages 382--401, 1999.

\bibitem{sheikhattar2015recursive}
A.~Sheikhattar, J.~B. Fritz, S.~A. Shamma, and B.~Babadi.
\newblock Recursive sparse point process regression with application to
  spectrotemporal receptive field plasticity analysis.
\newblock {\em arXiv preprint arXiv:1507.04727}, 2015.

\bibitem{shen2015influence}
L.~Shen, L.~Mickley, and A.~Tai.
\newblock Influence of synoptic patterns on surface ozone variability over the
  eastern united states from 1980 to 2012.
\newblock {\em Atmospheric Chemistry and Physics}, 15(19):10925--10938, 2015.

\bibitem{tibshirani1996regression}
R.~Tibshirani.
\newblock Regression shrinkage and selection via the lasso.
\newblock {\em Journal of the Royal Statistical Society. Series B
  (Methodological)}, pages 267--288, 1996.

\bibitem{wang2013accelerating}
Y.~Wang, M.~Li, and L.~Shen.
\newblock Accelerating carbon uptake in the northern hemisphere: evidence from
  the interhemispheric difference of atmospheric co 2 concentrations.
\newblock {\em Tellus B}, 65, 2013.

\bibitem{west1997time}
M.~West.
\newblock Time series decomposition.
\newblock {\em Biometrika}, 84(2):489--494, 1997.

\end{thebibliography}
\bibliographystyle{abbrv}

\end{document}